\documentclass{llncs}
\usepackage{fullpage}
\usepackage{amssymb}
\usepackage{verbatim}

\usepackage{hyperref}
\hypersetup{colorlinks=true,urlcolor=blue,linkcolor=blue,citecolor=blue}

\usepackage{tikz}
\usetikzlibrary{arrows}

\usepackage{algorithm}
\usepackage{algorithmic}
\usepackage{thm-restate}
\usepackage{subfigure}
\usepackage{xcolor}
\usepackage[outdir=./]{epstopdf}

\usepackage{amsmath}\allowdisplaybreaks

\newcommand{\compilefullversion}{true}
\ifthenelse{\equal{\compilefullversion}{false}}{%
	\newcommand{\OnlyInFull}[1]{}
	\newcommand{\OnlyInShort}[1]{#1}
}{%
	\newcommand{\OnlyInFull}[1]{#1}%
	\newcommand{\OnlyInShort}[1]{}%
}%

\newcommand{\compilehidecomments}{false}

\ifthenelse{ \equal{\compilehidecomments}{true} }{%
	\newcommand{\wei}[1]{}
	\newcommand{\ruihan}[1]{}
}{
	\newcommand{\wei}[1]{{\color{blue!50!black}  [\text{Wei:} #1]}}
	\newcommand{\ruihan}[1]{{\color{red!70!black} [\text{Ruihan:} #1]}}
}

\newcommand{\compileunmarkchanges}{false}
\ifthenelse{\equal{\compileunmarkchanges}{false}}{%
	\newcommand{\chgdel}[1]{{\color{violet}\sout{#1}}}%
}{%
	\newcommand{\chgdel}[1]{}%
}%

\newcommand{\E}{\mathbb{E}}
\newcommand{\I}{\mathbb{I}}
\newcommand{\R}{\mathbb{R}}

\newcommand{\vx}{\mathbf{x}}
\newcommand{\vy}{\mathbf{y}}
\newcommand{\ve}{\mathbf{e}}
\newcommand{\cX}{\mathcal{X}}
\newcommand{\cR}{\mathcal{R}}

\newcommand{\eg}{\hat{g}}
\newcommand{\evx}{\hat{\vx}}

\def\INPUT{\REQUIRE}
\def\OUTPUT{\ENSURE}
\newcommand{\List}{{\it List}}
\newcommand{\prev}{{\it prev}}
\newcommand{\ratio}{{\it ratio}}
\newcommand{\LB}{{\it LB}}
\newcommand{\OPT}{{\it OPT}}
\newcommand{\EPT}{{\it EPT}}

\newcommand{\LGreedy}{\mbox{\sf L-Greedy}}
\newcommand{\LGreedyDelta}{\mbox{\sf L-GreedyDelta}}

\newcommand{\Sampling}{{\sf Sampling}}
\newcommand{\NodeSelection}{{\sf NodeSelection}}
\newcommand{\IMM}{{\sf IMM}}

\newcommand{\IMMPRR}{\mbox{\sf IMM-PRR}}
\newcommand{\IMMVSN}{\mbox{\sf IMM-VSN}}

\newcommand{\UD}{{\sf UD}}
\newcommand{\CD}{{\sf CD}}
\newcommand{\HD}{{\sf HD}}
\newcommand{\MCLG}{{\sf MCLG}}

\newcommand{\argmax}{\operatornamewithlimits{argmax}}

\title{Scalable Lattice Influence Maximization}
\author{Wei Chen \inst{1} \and Ruihan Wu \inst{2} \and Zheng Yu \inst{3}}
\institute{Microsoft Research, Beijing, China \\ \email{weic@microsoft.com} \and Cornell University, Ithaca, NY, USA \\ \email{rw565@cornell.edu} \and
	Princeton University, Princeton, NJ, USA \\ \email{zhengy@princeton.edu}
}

\begin{document}

\maketitle

\begin{abstract}
Influence maximization is the task of finding $k$ seed nodes in a social network
such that the expected number of activated nodes in the network (under certain influence
propagation model), referred to as the influence spread, is maximized.
Lattice influence maximization (LIM) generalizes influence maximization such that, instead of
selecting $k$ seed nodes, one selects a vector $\vx=(x_1, \ldots, x_d)$ from a discrete
space $\cX$ called a lattice, where $x_j$ corresponds to the $j$-th marketing strategy and
$\vx$ represents a marketing strategy mix.
Each strategy mix $\vx$ has probability $h_u(\vx)$ to activate a node $u$ as a seed. 
LIM is the task of finding a strategy mix under the constraint $\sum_{j} x_j \le k$
such that its influence spread is maximized.
We adapt the reverse influence sampling (RIS) approach and design scalable algorithms for LIM.
We first design the {\IMMPRR} algorithm based on partial reverse-reachable sets as a general solution
for LIM, and improve {\IMMPRR} for a large family of models
where each strategy independently activates seed nodes.
We then propose an alternative algorithm {\IMMVSN} based on virtual strategy nodes, for the family
of models with independent strategy activations.
We prove that both {\IMMPRR} and {\IMMVSN} guarantees $1-1/e-\varepsilon$ approximation for
small $\varepsilon > 0$.
Empirically, through extensive tests we demonstrate that {\IMMVSN} runs faster than {\IMMPRR} 
and much faster than other baseline algorithms while providing the same level of influence
spread.
We conclude that {\IMMVSN} is the best one for models with independent strategy activations, while
{\IMMPRR} works for general modes without this assumption.
Finally, we extend LIM to the partitioned budget case where strategies are partitioned into groups, each of which has a separate budget, and
	show that a minor variation of our algorithms would achieve
	$1/2 -\varepsilon$ approximation ratio with the same time complexity.

%
\end{abstract}

\keywords{influence maximization, lattice influence maximization, 
		scalable influence maximization, reverse influence sampling}

\section{Introduction}

The classical influence maximization task is to find a small set of seed nodes to maximize the
	expected number of activated nodes from these seeds, referred to as the {\em influence spread},
	based on certain diffusion process
	in a social network~\cite{kempe03journal}.
It models the viral marketing scenario in social networks and its variants also find applications
	in diffusion monitoring, rumor control, crime prevention, etc.
	(e.g.,~\cite{Leskovec07,BAA11,HeSCJ12,ShakarianSPB14}).
Therefore, numerous studies on influence maximization have been conducted since its inception.
One important direction is scalable influence maximization, which aims at design efficient
	approximation algorithms and heuristics for large social networks.
Many diverse approaches including graph theoretic heuristics, sketching methods, and random sampling
	have been tried for scalable influence maximization (e.g.,~\cite{ChenWY09,simpath11,WCW12,BorgsBrautbarChayesLucier,tang14,tang15,CDPW14}).
Other directions include competitive and complementary influence
	maximization~\cite{BAA11,HeSCJ12,lu2015competition},
	continuous-time influence maximization~\cite{GomezRodriguez16},
	topic-aware influence maximization~\cite{ChenFLFTT15},
	etc.

However, a generalization of influence maximization already considered by Kempe et al.
	in their seminal paper~\cite{kempe03journal} receives much
	less attention and is left largely unexplored.
Kempe et al. consider viral marketing scenarios with a general marketing strategy mix
	of $d$ different strategies, with each strategy $j$ taking value $x_j$ (e.g., money
	put into strategy $j$).
The combined strategy mix is a vector $\vx = (x_1, x_2, \ldots, x_d)$.
When applying the strategy mix $\vx$ to the social network, each node $u$ in the social network
	has a probability of $h_u(\vx)$ to be activated as a seed.
After the seeds are probabilistically activated by the marketing strategies, influence propagates
	from the seeds in the network as dictated by an influence diffusion model.
The optimization problem is to find the best strategy mix $\vx^*$ that maximizes the influence spread
	subject to the budget constraint $\sum_{j\in [d]} x_j \le k$, where notation $[d]$ means $\{1,2,\ldots,d\}$.
In this paper, we consider strategy mixes taken from a discrete space $\cX$ referred to as a
	{\em lattice}, and thus we call the above optimization problem {\em lattice influence maximization (LIM)}.

LIM represents more realistic scenarios, since in practice
	companies often apply a mix of marketing strategies, such as coupons, direct mails, marketing events, and
	target at different segments of users.
In~\cite{kempe03journal}, Kempe et al. outline the basic approach based on submodularity and greedy
	algorithm to solve the problem.
This direction, however, has not been further investigated in the research community.
The only relevant study we find is~\cite{YangMPH16}, which investigates influence maximization with
	fractional or continuous discounts on users in the network, a special case of the LIM problem.

In this paper, we provide a detailed study on the scalable solutions for
	the LIM problem.
It is well known that the naive greedy approach for influence maximization is not scalable due to
	excessive Monte Carlo simulations.
The problem could be even worse for LIM when we have a large strategy space with complicated interactions
	with the social network.
We tackle this problem by adapting the reverse influence sampling (RIS)
	approach~\cite{BorgsBrautbarChayesLucier,tang14,tang15}, which is
	successful for the classical influence maximization problem.
The adaption of RIS to LIM is not straightforward, because nodes in the network are not deterministically
	selected as seeds but probabilistically selected based on the complicate function 
	$h_u(\vx)$.
In fact, the study in \cite{YangMPH16} does not apply the RIS approach and only provides some
	heuristic algorithms without any theoretical guarantee.
	
In our study, we first prove several important properties that enable the RIS approach in the LIM setting, 
	one of which in particular shows that the RIS approach in LIM can be interpreted as
	partial coverage of reverse-reachable sets.
From this we design a general scalable algorithm {\IMMPRR} adapted from the {\IMM} algorithm
	for the classical influence maximization problem~\cite{tang15}.

Then we identify a large class of LIM problems in which each strategy could independently activate
	nodes as seeds in the social network, which we call {\em independent strategy activations}.
We show that this class of problem covers many practical application scenarios including
	user segment marketing, personalized marketing, and repeated event marketing.
For this class of problems we revise {\IMMPRR} to improve its efficiency.
Next, we further investigate an alternative
	design choice where we convert strategies into virtual nodes so that the propagation can be
	reduced to that of the classical triggering model~\cite{kempe03}.
Again, although the idea of introducing virtual nodes seem to be natural, 
	it is nontrivial to make it exactly match the original LIM model.
In fact, we need to apply a novel integration of the classical independent cascade (IC) and 
	linear threshold (LT) models into a single model to make it work, and this integration
	could be of independent interest by itself.
We refer to the resulting algorithm as {\IMMVSN}.
For both {\IMMPRR} and {\IMMVSN}, we prove that they provide $1-1/e -\varepsilon$
	approximation to the LIM problem for any $\varepsilon > 0$, and we analyze their
	time complexity, which indicates that {\IMMVSN} could perform better in running time.

We conduct extensive experiments of our algorithms and several baseline algorithms 
	(including algorithms proposed in~\cite{YangMPH16})
	on four real-world networks with two different
	type of marketing strategies.
Our experimental results demonstrate that {\IMMVSN} is faster than {\IMMPRR}, and is much faster than
	all other baseline algorithms, while {\IMMVSN}/{\IMMPRR} provides the same or slightly better
	influence spread than other algorithms.
Moreover, for both {\IMMVSN}/{\IMMPRR}, we can easily tune one parameter to balance between
	theoretical guarantee and faster performance.

Finally, we generalize LIM originally proposed by \cite{kempe03journal} to accommodate
	partitioned budgets (denoted as the LIM-PB problem),
	that is, the strategies are partitioned into groups and each group has a
	separate budget.
This matches the practical scenario when marketing activities are coordinated by multiple parties, each
	of which focusing on different marketing channels with different marketing budgets.
We connect the LIM-PB problem with submodular maximization under matroid constraints, 
	and thus it implies that a minor variation of our algorithms would achieve
	$1/2 -\varepsilon$ approximation ratio with the same time complexity.

In summary, we make the following contributions: (a) we propose two scalable algorithms to solve
	the LIM problem with theoretical guarantees, one is more general and the other is more
	efficient in the case of independent strategy activations; 
	(b)
	we demonstrate through experiments that {\IMMVSN} is the best for the case of independent
	strategy activations, and runs much faster than other algorithms;
	and 
	(c) 
	we extend the problem of LIM to the case of partitioned budgets, and show that our
	scalable algorithms can still provide constant approximation.



\vspace{-2mm}
\subsection{Related Work}
\label{sec:relatedwork}

Influence maximization for viral marketing is first studied as a data mining task in
	~\cite{domingos01,richardson02}.
Kempe et al. \cite{kempe03journal} are the first to formulate the problem as a discrete optimization problem.
They propose the independent cascade (IC), linear threshold (LT), triggering, and other more general models,
	study their submodularity, and propose the greedy algorithm that gives
	$1-1/e-\varepsilon$ approximate solution for $\varepsilon > 0$.
They also propose the LIM problem and the greedy approach to solve the problem.

Scalable influence maximization is an important direction and receives many attention.
Some early proposals rely on the properties of the IC and LT models as well as efficient graph algorithms
	to design scalable heuristics~\cite{CYZ10,WCW12,simpath11,JungHC12}.
Borgs et al. \cite{BorgsBrautbarChayesLucier} propose the novel approach of reverse influence sampling (RIS), which is able to provide both theoretical guarantee and scalable performance in practice.
The RIS approach is improved by a series of studies~\cite{tang14,tang15,NguyenTD16,TangTXY18},
	which is also a demonstration that even with the known RIS approach achieving scalable
	influence maximization still requires significant design effort.
Our algorithm is based on RIS and is adapted from the {\IMM} algorithm~\cite{tang15}.
The adaptations from other algorithms (e.g.~\cite{NguyenTD16,TangTXY18}) would be similar, 
	and we choose {\IMM} mainly for its relative simplicity for illustrative purpose. 	

Many other directions of influence maximization have been studied, such as
	competitive and complementary influence
	maximization,
	seed minimization, etc.
They are less relevant to our study, so we refer to a monograph~\cite{chen2013information} for
more comprehensive coverage on influence maximization.

%
	
In terms of the LIM problem, the most relevant study is the one in~\cite{YangMPH16}.
In their model, each user could receive a personalized discount, which is translated to the probability
	of the user being activated as a seed.
This corresponds to the personalized marketing scenario in our setting.
They propose a scalable heuristic algorithm based on coordinate decent to solve the problem.
Comparing to their study, our algorithm is better in
	(a) providing theoretical guarantees
	on approximation ratio and running time;
	(b) solving a larger class of problems covering
	segment marketing, event marketing etc.;
	and (c) outperforming their algorithm
	in both running time and influence spread.
	
Demaine et al. propose a fractional influence model, in which the fractional solution $x_v$
	for a node $v$ affects not only on  $v$'s activation as a seed but also on $v$'s activation
	by its neighbors during the diffusion process
	\cite{DemaineHMMRSZ14}.
Thus, their model is incomparable with our LIM model, although both allow fractional solutions.

DR-submodular function maximization over lattices or continuous domain receives many attentions in
	recent years (e.g.,~\cite{FeldmanNS11,SomaKIK14,HassaniSK17}).
The main difference is that our algorithmic design focuses
	on the specific DR-submodular function related
	to the influence maximization task, while those studies focus on general
	DR-submodular functions.
Another difference is that they often rely on gradient methods, which assume that the function is continuous
	and differentiable, but we do not rely on such assumptions.
%
%
%
%

\section{Model and Problem Definition}

Influence propagation in social networks is modeled by the triggering model~\cite{kempe03journal}.
A social network is modeled as a directed graph $G=(V,E)$, where $V$ is the set of nodes representing
	individuals, and $E$ is the set of directed edges representing influence relationships.
We denote $n=|V|$ and $m=|E|$.
In the triggering model, every node $v$ has a distribution $D_v$ over all subsets of its in-neighbors.
Each node is either inactive or active, and once active it stays active.
Before the propagation starts, each node $v$ samples a triggering set $T_v\sim D_v$.
The propagation proceeds in discrete time steps $t=0, 1, 2, \ldots$.
At time $t=0$, nodes in a given {\em seed set} $S\subseteq V$ are activated.
For any time $t\ge 1$, an inactive node $v$ becomes active if any only if at least one of its
	in-neighbors in $T_v$ becomes active by time $t-1$.
The propagation ends when there is no newly activated nodes at a step.
Two classical models, independent cascade (IC) and linear threshold (LT), are both special cases
	of the triggering model:
	In the IC model, each edge $(u,v)$ has an influence probability $p(u,v)$, and the triggering set $T_v$
	is sampled by independently sample every incoming edge $(u,v)$ of $v$ with success probability
	$p(u,v)$ and put $u$ into $T_v$ if the edge sample is successful;
	in the LT model, each edge $(u,v)$ has an influence weight $w(u,v)\in [0,1]$ such that $\sum_{u} w(u,v)\le 1$, and at most one in-neighbor $u$ is sampled into $T_v$ with probability
	proportional to $w(u,v)$.
When considering time complexity, we assume that each sample $T_v$ can be drawn with time proportional
	to the in-degree of $v$, and this holds for both IC and LT models.

A key quantity is the {\em influence spread} of a seed set $S$, denoted as $\sigma(S)$, which is defined as
	the expected number of final active nodes for the propagation starting from $S$.
The classical influence maximization task is to select at most $k$ seed nodes to maximize the
	influence spread, i.e., to find $S^* = \arg\max_{S\subseteq V, |S|\le k} \sigma(S)$.
The problem is NP hard, and \cite{kempe03journal} proposes the greedy approximation algorithm, which is based
	on the submodularity of $\sigma(S)$ and guarantees $1-1/e - \varepsilon$ approximation
	for any small $\varepsilon > 0$.
	
In this paper we study the extension of
	influence maximization with general marketing strategies~\cite{kempe03journal}.
A mix of marketing strategies is modeled as a $d$-dimensional vector $\vx=(x_1, \ldots, x_d) \in \R_+^{d}$,
	where $\R_+$ is the set of nonnegative real numbers.
Each dimension corresponds to a particular marketing strategy, e.g., direct mail to one segment of the user base.
Under the marketing strategy mix $\vx$, each node $u\in V$ is independently activated as a seed with the probability given by the {\em strategy activation function} $h_u(\vx)$.
Then the set of activated seed nodes propagate the influence in the network following 
	the triggering model.
We define the {\em influence spread} of a marketing strategy mix $\vx$ as the expected number of nodes activated, and denote it as $g(\vx)$:
\begin{equation} \label{eq:vectorspread}
g(\vx) =\E_S[\sigma(S)] = \sum_{S\subseteq V}\sigma(S)\cdot\prod_{u\in S}h_u(\vx)\cdot\prod_{v\notin S}(1-h_v(\vx)).
\end{equation}
The above formula can be interpreted as follows: 
	for each subset of nodes $S$, under the marketing strategy mix $\vx$, the probability that exactly nodes in $S$ are activated as seeds
	and nodes not in $S$ are not activated as seeds is given as $\prod_{u\in S}h_u(\vx)\cdot\prod_{v\notin S}(1-h_v(\vx))$, which is because the node activations
	are independent.
Then given that exactly nodes in $S$ are activated as seeds, the influences spread it generates is $\sigma(S)$.
Therefore, enumerating through all possible subset set $S$, we obtain the above formula.

In this paper, we consider discretized marketing strategies with granularity parameter $\delta$,
	i.e., each strategy $x_i$ takes discretized values $0, \delta, 2\delta, \ldots$.
These set of vectors is referred to as a lattice, and is denoted as $\cX$.
We consider the marketing strategy mix $\vx$ with a total budget constraint $k$: $|\vx| \le k$,
	where $|\vx| = \sum_{i\in [d]}  x_i$.
The above constraint can be thought as the total monetary budget constraint, where $x_j$ is the monetary expense on strategy $j$, but other interpretations are also possible.
Since we are doing influence maximization on lattice $\cX$, we call it {\em lattice influence maximization}, as formally defined below.

\begin{definition}[Lattice Influence Maximization]
	\label{def:LIM}
Given a social network $G=(V,E)$ with the triggering model parameters $\{D_v\}_{v\in V}$, given the strategy activation functions $\{h_v\}_{v\in V}$
	and a total budget $k$,
	the task of Lattice influence maximization, denoted as LIM, is to find an optimal strategy mix $\vx^*$ that achieves
	the largest influence spread within the budget constraint, that is
\[
	\vx^* = \argmax_{\vx\in \cX, |\vx|\le k } g(\vx).
\]

\end{definition}

Note that if $\cX = \{0,1\}^{n}$ and $h_v(\vx) = x_v$, that is, $v$ is activated as a seed if and only if it is selected by strategy $\vx$, the LIM problem becomes the classical influence maximization problem.
Therefore, LIM is more general, and inherits the NP-hardness of the classical problem.
For convenience, we sometimes also use LIM to refer to the
	lattice-based propagation model described above.

To solve the LIM problem, \cite{kempe03journal} proposes the greedy algorithm based on the diminishing return property of $g(\vx)$, commonly
	referred to as the DR-submodular property~\cite{SomaY15}.
For two vectors $\vx, \vy \in \R^d$, we denote $\vx \le \vy$ if $x_i \le y_i$ for all $i\in [d]$.
Let $\ve_i \in \R^d$ be the unit vector with the $i$-th dimension being $1$ and all other
	dimensions being $0$.
For a vector function $f: \cX \rightarrow \R$, we say that $f$ is {\em DR-submodular} if
	for all $\vx,\vy \in \cX$ with $\vx \le \vy$, for all $i\in [d]$,
	$f(\vx + \delta\ve_i) - f(\vx) \ge f(\vy + \delta \ve_i) - f(\vy)$; and
	we say that $f$ is {\em monotone (nondecreasing)} if for all $\vx \le \vy$, $f(\vx) \le f(\vy)$.
Note that a set function $f$ is monotone if $f(S)\le f(T)$ for all $S\subseteq T$, and submodular
	if $f(S\cup \{u\}) - f(S) \ge f(T\cup \{u\}) - f(T)$ for all $S\subseteq T$ and $u\not\in T$.
It is clear that if we represent sets as binary vectors and take step size $\delta =1$, then it
	coincides with monotonicity and DR-submodularity of vector functions.

\begin{algorithm}[t]
	\caption{Algorithm $\LGreedy(f, k, \delta)$}
	\label{alg:hillclimb}
	\begin{algorithmic} [1]
		\INPUT{monotone DR-submodular $f$, budget $k$, granularity $\delta$}
		\OUTPUT{vector $\vx$}
		\STATE $\vx = \mathbf{0}$
		\FOR{$t = 1, 2, \cdots, k\cdot\delta^{-1}$}
			\STATE $j^* = \argmax_{j\in [d]} f(\vx + \delta \ve_j)$ \label{line:argmaxgreedy}
			\STATE $\vx = \vx + \delta \ve_{j^*}$
		\ENDFOR
		\STATE \textbf{return} $\vx$
	\end{algorithmic}
\end{algorithm}

When the vector function $f$ on lattice $\cX$ is nonnegative,
	monotone and DR-submodular, the lattice-greedy (denoted as {\LGreedy})
	algorithm as given in Algorithm~\ref{alg:hillclimb} achieves $1-1/e$
	approximation~\cite{NWF78}.
The {\LGreedy} algorithm searches the coordinate that gives the largest marginal return and
	moves one step of size $\delta$ on that coordinate, until it exhausts the budget.
%

To apply the $\LGreedy$ algorithm to LIM, \cite{kempe03journal} shows that
	when $h_v$'s are monotone and DR-submodular with $\sigma(S)$ being monotone and submodular,
	the influence spread $g(\vx)$ given in Eq.~\eqref{eq:vectorspread} is also monotone and DR-submodular.
Therefore, the $\LGreedy$ algorithm can be applied to $g(\vx)$.
As it is \#P-hard to compute the influence spread $\sigma(S)$ in the IC and LT models~\cite{WCW12,CYZ10},
	we could use Monte Carlo simulations to estimate $g(\vx)$ to achieve
	$1-1/e-\varepsilon$ approximation for any small $\varepsilon > 0$.
	
We remark that in the LIM problem, 
	for each strategy $j$, we can add an upper bound constraint $x_i \le b_i$
	without changing the problem, because we can extend the domain of $x_i$ beyond $b_i$ by
	restricting $h_v(\vx)$ with some $x_i > b_i$ to be the value at the boundary $x_i = b_i$.
It is easy to verify that this extension will not affect monotonicity and DR-submodularity of function
	$h_v$, nor will it affect the lattice-greedy algorithm.

\section{Scalable Algorithms for LIM}

It is well known that the Monte Carlo greedy algorithm is not scalable.
In this paper, we propose scalable solutions to the LIM problem based on the
	seminal reverse influence sampling (RIS) approach~\cite{BorgsBrautbarChayesLucier,tang14,tang15}.
In particular, we adapt the {\IMM} (Influence Maximization with Martingales) algorithm
	of~\cite{tang15} in two different ways, one relies on partial reverse reachable sets
	and is denoted as {\IMMPRR}, and the other uses virtual strategy nodes and is denoted
	as {\IMMVSN}.

\subsection{Reverse Reachable Sets and Its Properties}
The RIS approach is based on the key concept of the {\em reverse reachable sets (RR sets)}, as defined below.
\begin{definition}[Reverse Reachable Set]
Under the triggering model, a {\em reverse reachable (RR) set}
	rooted at a node $v$, denoted $R_v$, is the random set of nodes $v$ reaches in one reverse
	propagation: sample all triggering sets $\{T_u\}_{u\in V}$, such that edges 
	$\{(w,u) \mid u\in V, w\in T_u\}$ together with nodes $V$ form a live-edge graph, and
	$R_v$ is the set of nodes that can reach $v$ (or $v$ can reach reversely) in this live-edge graph.
An RR set $R$ without specifying a root is one with root $v$ selected uniformly at random from $V$.
\end{definition}
Intuitively, RR sets rooted at $v$ store nodes that are likely to influence $v$.
Technically, it has the following important connection with the influence spread of a seed set $S$:
	$\sigma(S) = n \cdot \E_R[\I\{S\cap R \neq \emptyset\}]$, where $\I$ is the indicator function
	\cite{BorgsBrautbarChayesLucier,tang15}.

For our LIM problem, our first key observation is that the above property can be extended
	in the following way as a partial coverage on RR sets.
\begin{lemma} \label{lem:RRSetGMS}
For any strategy mix $\vx \in \cX$, we have
\begin{equation} \label{eq:RRSetGMS}
g(\vx) = n \cdot \E_{R}\left[1 - \prod_{v\in R}(1 - h_v(\vx))\right].
\end{equation}
\end{lemma} 
\begin{proof}
By Eq.~\eqref{eq:vectorspread}, we have $g(\vx) = \E_{S}[\sigma(S)]= n\cdot \E_{S,R}[\I\{S\cap R \neq \emptyset\}] = n\cdot \E_R[\Pr_S\{S\cap R \ne \emptyset \} ]$.
Then $\Pr_S\{S\cap R \ne \emptyset \}$ is the probability that at least one node in $R$ (now fixed)
	is activated as a seed under strategy mix $\vx$, so it is
	$ 1 - \prod_{v\in R}(1 - h_v(\vx))$.
\hfill $\square$
\end{proof}

Lemma~\ref{lem:RRSetGMS} indicates that an RR set $R$ 
	is {\em partially} covered by a strategy mix $\vx$ with probability (or weight)
	$ 1 - \prod_{v\in R}(1 - h_v(\vx))$, instead of the classical case where
	an RR set is either fully covered by a seed set $S$ or not.
This lead to the partial RR set extension of IMM, called {\IMMPRR}.

\subsection{Algorithm {\IMMPRR}}
\label{sec:IMMPRR}

\noindent
{\bf General Structure of {\IMMPRR}.\ \ }
By Eq.\eqref{eq:RRSetGMS}, we can generate $\theta$ independent RR sets as a collection $\cR$ to
	obtain
\begin{equation} \label{eq:hatgdef}
	\eg_{\cR}(\vx) = \frac{n}{\theta}\sum_{R\in \cR} \left( 1 - \prod_{v\in R}(1 - h_v(\vx)) \right)
\end{equation}
	as an unbiased estimate of $g(\vx)$.
Moreover, we have the following property for $\eg_{\cR}(\vx)$.
\begin{restatable}{lemma}{lemhtogsubmodular} \label{lem:htogsubmodular}
If $h_v$ is monotone and DR-submodular for all $v\in V$, then functions $g$ and 
	$\eg_{\cR}$ are also monotone and DR-submodular.	
\end{restatable}
\begin{proof}[Sketch]
	We apply the technical Lemma~\ref{lem:productSubmodular} below on
	$\prod_{v\in R}(1 - h_v(\vx))$, and notice that
	$1 - h_v(\vx)$ is nonnegative, monotone nonincreasing, and DR-supermodular.
	Therefore, $1 - \prod_{v\in R}(1 - h_v(\vx))$ is nonnegative, monotone increasing, and DR-submodular.
\hfill $\square$
\end{proof}

\begin{lemma} \label{lem:productSubmodular}
	If $f_1$ and $f_2$ are nonnegative, monotone nonincreasing and DR-supermodular, then $f(\vx) = f_1(\vx) f_2(\vx)$ is also monotone nonincreasing and DR-supermodular.
\end{lemma}
\begin{proof}
	The monotonicity is straightforward. For DR-supermodularity, for any $\vx \le \vy$, we have
	\begin{align*}
	&(f(\vx+ \delta\ve_i) - f(\vx)) - (f(\vy + \delta \ve_i) - f(\vy) )\\
	& = f_1(\vx+ \delta\ve_i)f_2(\vx+ \delta\ve_i) - f_1(\vx)f_2(\vx) \\
	& \quad - (f_1(\vy+ \delta\ve_i)f_2(\vy+ \delta\ve_i) - f_1(\vy)f_2(\vy) ) \\
	& = f_1(\vx+ \delta\ve_i)(f_2(\vx+ \delta\ve_i)-f_2(\vx))  \\
	& \quad	+ f_2(\vx)(f_1(\vx + \delta \ve_i) - f_1(\vx))  \\
	& \quad  - f_1(\vy+ \delta\ve_i)(f_2(\vy+ \delta\ve_i)-f_2(\vy)) \\
	& \quad - f_2(\vy)(f_1(\vy + \delta \ve_i) - f_1(\vy))  \\
	& \le f_1(\vx+ \delta\ve_i)(f_2(\vy+ \delta\ve_i)-f_2(\vy)) \\
	& \quad + f_2(\vx)(f_1(\vy + \delta \ve_i) - f_1(\vy))  \\
	& \quad - f_1(\vy+ \delta\ve_i)(f_2(\vy+ \delta\ve_i)-f_2(\vy)) \\
	& \quad - f_2(\vy)(f_1(\vy + \delta \ve_i) - f_1(\vy))  \\
	& = (f_1(\vx+ \delta\ve_i) - f_1(\vy+ \delta\ve_i)  )(f_2(\vy+ \delta\ve_i)-f_2(\vy)) \\
	& \quad + (f_2(\vx) - f_2(\vy) ) (f_1(\vy + \delta \ve_i) - f_1(\vy)) \le 0,
	\end{align*}
	where the first inequality is due to the DR-supermodular and nonnegative conditions, and the second inequality is
	due to the monotone nonincreasing property.
\end{proof}

With Lemma~\ref{lem:htogsubmodular}, we can apply the $\LGreedy$ algorithm on $\eg_{\cR}$.
Let $\evx^o = \LGreedy(\eg_{\cR}, k, \delta)$.
When $\theta = |\cR|$ is large enough, $\eg_{\cR}$ is very close to $g$, and we could show that
	$\evx^o$ is a $1-1/e -\varepsilon$ approximation for the LIM problem.

\begin{algorithm}[t]
	\caption{General structure of \IMMPRR}
	\begin{algorithmic} [1]
		\INPUT{$G$: the social graph; $\{D_v\}_{v\in V}$: triggering model parameters;
			$\{h_v\}_{v\in V}$: strategy activation functions (or $\{q_{v,j}\}_{v\in V, j\in S_v}$ for
			$\LGreedyDelta$;)
		$k$: budget;
		 $\delta$: granularity; $\varepsilon$: accuracy; $\ell$: confidence}
		\OUTPUT{$\vx \in \cX$}
		\STATE $\cR=\Sampling(G, \{D_v\}_{v\in V}, k, \delta, \varepsilon, \ell)$
		\STATE $\vx =\LGreedy(\eg_{\cR}, k, \delta)$ \\
			// or $\LGreedyDelta(\cR, \{q_{v,j}\}_{v\in V, j\in S_v}, k, \delta)$
		\STATE \textbf{return} $\vx$
	\end{algorithmic}
	\label{alg:IMMGMS}
\end{algorithm}

This leads to the general structure of the {\IMMPRR} algorithm as given in Algorithm~\ref{alg:IMMGMS},
	similar to the {\IMM} algorithm. 
The algorithm takes the input as listed in Algorithm~\ref{alg:IMMGMS}
	and outputs
	$\vx$ such that $\vx$ is a $1-1/e-\varepsilon$ approximate solution to the LIM
	problem with probability at least $1-n^\ell$.
The algorithm contains two phases.
In the first phase, the $\Sampling$ procedure determines the number of RR sets needed and
	generates these RR sets; in the second phase, a lattice-greedy algorithm on these RR sets are
	used to find the resulting strategy vector $\vx$.
We first discuss the second phase, which requires major changes from the original {\IMM} algorithm,
	and then introduce the first phase.

\vspace{2mm}
\noindent
{\bf Efficient {\LGreedy} on RR Sets under Independent Strategy Activation.\ \ }
\label{sec:application}
If the strategy activation function $h_v(\cdot)$'s are given as black boxes,
	we have to compute $h_v(\vx)$ from scratch.
Suppose that the running time cost for computing $h_v(\vx)$ is $O(T_{h_v})$.
Then it is straightforward to verify that
	the $\LGreedy(\eg, k, \delta)$ algorithm with the computation of $\eg(\vx)$ as given
in Eq.~\eqref{eq:hatgdef} has time complexity  $O(k \cdot \delta^{-1}\cdot d\cdot \sum_{R\in \cR}\sum_{v\in R}T_{h_v})$.

When we have further structural knowledge about $h_v$'s, we can greatly improve the efficiency
	of the {\LGreedy} algorithm.
In particular, we consider a large class of functions where each strategy $j$
	independently try to activate $v$ as a seed.
We refer to this case as {\em independent strategy activation}.
Suppose that the set of strategies that may activate $v$ is $S_v\subseteq [d]$, and
	the probability that strategy $j$ with amount $x_j$ activates $v$ as a seed is $q_{v,j}(x_j)$,
	with $q_{v,j}(0)=0$.
Then we have
\begin{equation}
h_v(\vx) = 1 - \prod_{j\in S_v}(1 - q_{v,j}(x_j)).
\label{eq:hv}
\end{equation}
We assume that $q_{v,j}(x)$ is non-decreasing and concave for every $j\in S_v$.
The following lemma shows that in this case $h_v(\vx)$ is monotone and DR-submodular.
\begin{restatable}{lemma}{lemDRSubmodularqtoh} \label{lem:DRSubmodularq2h}
	If function $q_{v,j}(x)$ is non-decreasing and concave for every $j\in S_v$, then
	$h_v(\vx)$ is monotone and DR-submodular.	
\end{restatable}
\begin{proof}[Sketch]
	The proof also uses Lemma~\ref{lem:productSubmodular}, and we only need to notice that one-dimensional convexity is
	a special case of DR-supermodularity.
	%
\hfill $\square$
\end{proof}

We now justify the independent strategy activation assumption (Eq.~\eqref{eq:hv}) 
	with several application scenarios.
The first application scenario is \textit{user segment marketing}, in which each strategy $j$  targets at a disjoint subset of users $V_j$.
In this case, for each user $v$, it has a unique strategy targeted at $v$, i.e. $|S_v|=1$.

The second scenario is \textit{personalized marketing}, where each user is targeted with a personalized strategy.
The personalized discount strategies studied in~\cite{YangMPH16}
belongs to this scenario.
Technically, this scenario is a special case of the above segment marketing scenario, where the user segments $V_j$'s are all singletons,
and $d = n$.


The third scenario is repeated marketing such as \textit{multi-event marketing}.
For example, each strategy $j$ is a type of events, and $x_j$ is the number of events of type $j$.
Suppose that for each event of type $j$, a user $v$ targeted by this event has an independent probability $r_{v,j}$ to be activated as a seed,
then $q_{v,j}(x_j) = 1 - (1 - r_{v, j})^{x_j}$.
This is a concrete example where $q_{v,j}(x)$ is non-decreasing and concave, and thus by Lemma~\ref{lem:DRSubmodularq2h}	$h_v(\vx)$ is monotone and DR-submodular.

Eq.~\eqref{eq:hv} enables more efficient updates for {\LGreedy}:
	Instead of always computing $\eg_\cR(\vx+\delta\ve_j)$ from scratch in
	$\LGreedy(\eg, k, \delta)$, we compute $\Delta_j(\vx) = \eg_\cR(\vx + \delta\ve_j) - \eg_\cR(\vx)$,
	which is given below.
$$
\Delta_j(\vx) = \frac{n}{\theta}\sum_{R\in \cR}\left(\prod_{v\in R}\prod_{j'\in S_v}(1 - q_{v,j'}(x_{j'}) )\right)\cdot
$$
\begin{equation}
 \left( 1 - \frac{\prod_{v: v\in R, j\in S_v} (1 - q_{v,j}(x_j + \delta))}{\prod_{v: v\in R,  j\in S_v} (1 - q_{v,j}(x_j) )}\right).
\label{eq:delta}
\end{equation}

\begin{algorithm}[t]
	\caption{\LGreedyDelta: Efficient lattice-greedy implementation on RR sets}
	\begin{algorithmic} [1]
		\INPUT{$\cR=\{R_1, \ldots, R_\theta \}$: RR sets;
			$\{q_{v,j}\}_{v\in V, j\in S_v}$; 
			$k$: budget; $\delta$: granularity}
		\OUTPUT{$\vx \in \cX$}
		\STATE $\vx = (x_1, \cdots, x_d) = \mathbf{0}$
		\STATE // Lines~\ref{line:initializes}--\ref{line:initializeList} can be done
			while generating RR sets
		\STATE \label{line:initializes}
			$\mathbf{s} = (s_0, s_1, \cdots, s_{\theta})$ with $s_0=0$, $s_i = \prod_{v\in R_i}\prod_{j\in S_v}  (1 - q_{v,j}(x_j) )$
		\STATE $\forall j\in [d]$, $\List_j = \emptyset$
		\STATE \label{line:initializeList}
		$\forall R_i \in \cR$, $\forall v\in R_i$, $\forall j\in S_v$,
			append $(i,v)$ to $\List_j$
		\FOR{$t = 1, 2, \cdots, k\cdot\delta^{-1}$}
		\FOR{$j\in [d]$} \label{line:iteratejb}
		\STATE $\Delta_j = 0$, $\prev = 0$, $\ratio = 1$
		\FOR{$(i, v)\in \List_j$} \label{line:forListb}
		\IF{$i \neq \prev$}
		\STATE$\Delta_j = \Delta_j + s_{\prev}\cdot (1 - \ratio)$
		\STATE $\ratio = 1$
		\STATE $\prev = i$
		\ENDIF
		\STATE $\ratio = \ratio \cdot \frac{1- q_{v,j}(x_j + \delta)}{1 - q_{v,j}(x_j)}$
		\ENDFOR \label{line:forListe}
		\IF{$\prev \neq 0$}
		\STATE $\Delta_j = \Delta_j + s_{\prev}\cdot (1 - \ratio)$
		\ENDIF
		\ENDFOR \label{line:iterateje}
		\STATE $j^* = \argmax_{j\in [d]} \Delta_{j}$ \label{line:largestj}
		\STATE $\vx= \vx + \delta \mathbf{e}_{j^*}$ \label{line:movej}
		\STATE $\forall i\in [\theta]$ , $s_i = s_i \cdot \prod_{v\in R_i: j^*\in S_v} (1 - q_{v, j^*} (x_{j^*} + \delta)) \cdot (1 - q_{v, j^*}(x_{j^*}))^{-1}$ \label{line:updates}
		\ENDFOR
		\STATE \textbf{return} $\vx$
	\end{algorithmic}
	\label{alg:hill_climbing_delta}
\end{algorithm}

The advantage of Eq.~\eqref{eq:delta} is in reusing past computations.
Specifically, the term within the first parentheses is the same across all strategies, so its
	computation can be shared.
Moreover, since it is often the case that each user is only exposed to a small subset of strategies
	(i.e. $|S_v|$ is smaller than $d$), we carefully maintain a data structure to improve the efficiency
	when $|S_v| < d$.
Algorithm \ref{alg:hill_climbing_delta} presents the detailed lattice-greedy update procedure
	$\LGreedyDelta$, which replaces
	$\LGreedy(\eg, k, \delta)$ when Eq.~\eqref{eq:hv} holds.

In Algorithm \ref{alg:hill_climbing_delta}, we use $s_i$ to store the term
	$\prod_{v\in R_i}\prod_{j'\in S_v} (1 - q_{v,j'}(x_{j'}) )$ in Eq.~\eqref{eq:delta} shared across different
	strategies $j$.
We use $\ratio$ to store the ratio term $\prod_{v: v\in R, i\in S_v} (1 - q_{v,i}(x_i + \delta))(1 - q_{v,i}(x_i))^{-1} $
	in Eq.~\eqref{eq:delta}.
The $\List_j$ is a linked list for strategy $j$, and it stores the pair $(i,v)$, which
	means RR set $R_i$ contains node $v$ that can be affected by strategy $j$.
The list is ordered by RR set index $i$ first and then by node index $v$.
In each round $t$, the algorithm iterates through all strategies $j$
	(lines~\ref{line:iteratejb}--\ref{line:iterateje}) to compute $\Delta_j(\vx)$ for the current $\vx$.
In particular, for each strategy $j$, the algorithm traverses the $\List_j$
	(lines~\ref{line:forListb}--\ref{line:forListe}), and for the segment with the same RR set index
	$i$, it updates $\ratio$, and when it reaches a new RR set index ($i\ne \prev$), it cumulates $\Delta_j$
	as given in Eq.~\eqref{eq:delta} for the corresponding RR set.
The reason we maintain $\List_j$ of pairs instead of simply looping through all RR set indices $i$
	and then all nodes within $R_i$ is that RR sets are usually not very large, and it is likely
	that no node in RR set $R_i$ is affected by strategy $j$, and thus not looping through all
	RR sets save time.
After computing $\Delta_j = \Delta_j(\vx)$, we find the strategy $j^*$ with the largest $\Delta_j$
	(line~\ref{line:largestj}), move along the direction of $j^*$ for one step (line~\ref{line:movej}), and
	then update all shared terms $s_i$'s (line~\ref{line:updates}).

Suppose that the running time cost for computing each $q_{v,j}(x_j)$ is a constant.
Then we have:
\begin{restatable}{lemma}{lemtimeDelta} \label{lem:timeDelta}
The time complexity of $\LGreedyDelta$
	is $O( k \cdot \delta^{-1}\cdot (\sum_{R\in \cR}\sum_{v\in R} |S_v|))$.
\end{restatable}
\begin{proof}[Proof of Lemma~\ref{lem:timeDelta} (Sketch)]
	The algorithm has totally $k \delta^{-1}$ rounds.
	In each round, it enumerates all tuples $(i,v,j)$ for RR set $R_i$, node $v\in R_i$ and
	strategy $j \in S_v$, and for each tuple it has a constant number of calls to function $q_{v,j}$,
	so the running time in one round $t$ is $O(\sum_{R\in \cR}\sum_{v\in R} |S_v|)$.
\hfill $\square$
\end{proof}

Notice that if we compute $\eg(\vx + \delta \ve_j)$ directly instead of $\Delta_j(\vx)$,
	we have $T_{h_v} = O(d)$.
Then time complexity is
	$O(k \cdot \delta^{-1} \cdot d \cdot \sum_{R\in \cR}\sum_{v\in R} |S_v|)$, which is worse than
	$\LGreedyDelta$ by a factor of $d$.


\vspace{2mm}
\noindent
{\bf The First Phase $\Sampling$ Procedure.\ \ }
The $\Sampling$ procedure in the first phase is to generate enough RR sets $\cR$ to provide
	the theoretical guarantee on the approximation ratio.
It is a minor variation of the $\Sampling$ procedure of {\IMM} in~\cite{tang15}.
In particular, they show that the number of RR sets
	$\theta = \Theta(n \log n / \OPT)$ is enough, where $\OPT$ is the optimal solution.
They estimate a lower bound $\LB$ of $\OPT$ by iteratively guessing
	$n/2, n/4, n/8, \ldots$ as lower bounds, and using the greedy procedure on obtained
	RR sets to verify if the guess is correct.
We use the same procedure, with only two differences:
	(a) we use $\LGreedyDelta$ procedure to replace the greedy procedure on RR sets; and
	(b) we replace $\ln \binom{n}{k}$ with $\min(k \delta^{-1} \ln d, d \ln(k\delta^{-1}))$ in the two parameters
	$\lambda'$ and $\lambda^*$, because both $d^{k\delta^{-1}}$ and $(k\delta^{-1})^d$ are
	upper bounds on the number of vectors satisfying the constraint $|\vx| \le k$.
The bound $d^{k\delta^{-1}}$ is because we have $k\delta^{-1}$ greedy steps and each step selects
	one dimension among $d$ dimensions, and the bound $(k\delta^{-1})^d$ is because
	each dimension has at most $k\delta^{-1}$ choices and we have $d$ dimensions combined together.
We can see that when $d$ is large (e.g. personalized marketing with $d=n$) but $k\delta^{-1}$ is relatively small
	(coarse granularity), we would use $k\delta^{-1} \ln d$, but when $k\delta^{-1}$ is large (fine granularity) but $d$ is small (e.g. only a few global strategies), we could use 
	$d \ln(k\delta^{-1})$.
Henceforth, we let $M = \min(k \delta^{-1} \ln d, d \ln(k\delta^{-1}))$.
The pseudocode for the $\Sampling$ procedure is included in Algorithm~\ref{alg:sampling}, 
	with parameter $\lambda^*(\ell)$ defined below. 
\begin{align}
	&\lambda^*(\ell)  = 2n\cdot \left((1-1/e)\cdot\alpha + \beta\right)^2\cdot\varepsilon^{-2},
	\label{eq:lambda*} \\
	&\alpha = \sqrt{\ell\ln n + \ln 2},  
	\beta = \sqrt{(1-1/e)\cdot\left(M + \alpha^2 \right)}. \nonumber
\end{align}
\begin{algorithm}[t]
	\caption{First phase $\Sampling$ procedure}
	\begin{algorithmic} [1]
		\INPUT{$G$: the social graph; $\{D_v\}_{v\in V}$: triggering model parameters; 
			$\{q_{v,j}\}_{v\in V, j\in S_v}$: strategy-node activation functions; 
			$k$: budget, $\delta$: granularity; $\varepsilon$: accuracy; $\ell$: confidence}
		\OUTPUT{A collection of RR sets $\cR$}
		\STATE $\cR = \emptyset$; $\LB = 1$
		\STATE \label{alg:workaroundb}
		compute $\gamma$ via binary search such that 
		$\lceil\lambda^*(\ell + \gamma) \rceil /n^{\ell+\gamma} \le 1/n^\ell$ 
		// workaround 2 in~\cite{Chen18}, with $\lambda^*(\ell)$ defined in Eq.~\eqref{eq:lambda*}
		\STATE  \label{alg:workarounde} $\ell=\ell+\gamma + \ln 2 / \ln n$
		\STATE Let $\varepsilon' = \sqrt{2}\cdot \varepsilon$
		\FOR{$i=1, 2, \cdots, \log_2 n$} \label{line:forestimateb}
		\STATE Let $y = n/2^i$
		\STATE $\theta_i = \frac{\lambda'}{y}$, where $\lambda' = \frac{\left(2 + \frac{2}{3}\varepsilon'\right)\cdot\left(M + \ell\cdot\ln n + \ln\log_2n\right)\cdot n}{\varepsilon'^2}$.
		\WHILE{$|\cR|\leq \theta_i$}
		\STATE Select a node $v$ from $G$ uniformly at random
		\STATE Generate an RR set for $v$, and insert it into $\cR$
		\ENDWHILE
		\STATE $\vx =\LGreedyDelta(\cR, \{q_{v,j}\}_{v\in V, j\in S_v}, k, \delta)$
		\label{line:callhillclimb}
		\IF{$\eg_\cR(\vx) \geq (1+\varepsilon')\cdot y$} \label{line:verifyifb}
		\STATE $\LB = \eg_\cR(\vx )/ (1+\varepsilon')$
		\STATE \textbf{break}
		\ENDIF \label{line:verifyife}
		\ENDFOR \label{line:forestimatee}
		\STATE $\theta = \lambda^*(\ell) / \LB$, where $\lambda^*(\ell)$ is defined in Eq.~\eqref{eq:lambda*}
		\label{line:settheta}
		\WHILE{$|\cR|\leq\theta$}
		\STATE Select a node $v$ from $G$ uniformly at random
		\STATE Generate an RR set for $v$, and insert it into $\cR$
		\ENDWHILE
		\STATE \textbf{return} $\cR$
	\end{algorithmic}
	\label{alg:sampling}
\end{algorithm}

We remark that Chen pointed out an issue in the original {\IMM} algorithm 
and provided two workarounds~\cite{Chen18}, and we adopt the more efficient
workaround 2 (lines~\ref{alg:workaroundb}-\ref{alg:workarounde}).
Algorithms~\ref{alg:IMMGMS}, \ref{alg:hill_climbing_delta}, and~\ref{alg:sampling}
	form the {\IMMPRR} algorithm.
The following theorem summarizes the theoretical guarantee of the
	{\IMMPRR} algorithm.

\begin{restatable}{theorem}{thmIMMGMS}
Under the case of independent strategy activation (Eq.~\eqref{eq:hv}),
	the {\IMMPRR} algorithm
	returns a $(1-1/e-\varepsilon)$-approximate solution
	to the LIM problem with at least $1 - 1/n^\ell$ probability.
When $q_{v,j}$'s are such that the
	optimal solution of LIM
	is at least as good as the best single node influence spread, {\IMMPRR} runs in $O(k\delta^{-1}(\max_{v\in V}|S_v|)(M+\ell\log n)(n + m)/  \varepsilon^2)$ expected time,
	where $M= \min(k \delta^{-1} \ln d, d \ln(k\delta^{-1}))$.
\label{thm:IMMGMS}
\end{restatable}

The proof of the theorem mainly follows the analysis of {\IMM} in \cite{tang15}, and 
	the novel part of the analysis is already mostly shown in the previous lemmas.
The remaining part of the proof is given in Appendix~\ref{app:proofThm1}.
Note that the technical assumption above assuming the optimal solution is at least as good as the best single node influence spread is reasonable, since
	it means the budget and the functions $q_{v,j}$'s are at least good enough to activate one single best node.
If it is not true, the entire marketing scheme is not very useful anyway.
Comparing to the time complexity $O((k+\ell)(m+n) \log n /\varepsilon^2)$ of {\IMM} in~\cite{tang15},
	the main added difficulty is that a strategy can only partially cover an RR set
	(Lemma~\ref{lem:RRSetGMS}), which implies that in each greedy step we have to
	process all RR sets.
We will overcome this issue by an alternative reduction approach in the next subsection.

\subsection{Algorithm {\IMMVSN} for Independent Strategy Activation}

In this subsection, we consider an alternative design choice under independent strategy activation.
The idea is that since each strategy independently activates nodes, we may be able to 
	introduce virtual nodes representing strategies such that the LIM model is
	reduced to the classical triggering model, and then 
	we could apply algorithms such as IMM to solve
	the classical influence maximization problem under the reduced model.
It turns out that we need to incorporate a mixture of LT and IC models for the interaction between the virtual nodes
	and the real nodes, and carefully argue about the equivalence between LIM and the reduced
	model.
We refer this new algorithm as {\IMMVSN} (VSN stands for virtual strategy nodes).
	
In {\IMMVSN}, for every strategy $j$, we construct virtual strategy node set 
	$U_j = \{u_{j,1}, u_{j,2}, \ldots, u_{j,k\delta^{-1}}\}$, 
	and for every real node $v$ in the original graph and every strategy $j\in S_v$, 
	we connect every virtual node $u_{j,i}$ to $v$
	with a directed virtual edge $(u_{j,i},v)$.
Let $U = \bigcup_{j=1}^d U_j$ be the set of all virtual nodes.
The purpose is such that the prefix set $U_{j,i}=\{u_{j,1}, \ldots, u_{j,i}\}$ 
	corresponds to the quantity $x_j = i\delta$
	for strategy $j$, and if nodes in $U_{j,i}$ are seeds, 
	then real node $v$ is activated with probability $q_{v,j}(i\delta)$, the probability that 
	amount $i\delta$ of strategy $j$ would activate $v$ 
 	(see Eq.\eqref{eq:hv}).
To do so, we utilize the LT model as follows.
For each edge $(u_{j,i},v)$, we assign LT weight 
\begin{equation}  \label{eq:LTweight}
	w(u_{j,i},v) = q_{v,j}(i\delta) -  q_{v,j}((i-1)\delta).
\end{equation}
When a seed set $S\subseteq U$ of virtual nodes attempts to activate a real node $v$, 
	we first consider seed set within each strategy $S\cap U_j$, and nodes in  $S\cap U_j$ 
	attempt to activate $v$
	following the LT model with weights defined in Eq.~\eqref{eq:LTweight}.
Then among different strategies, their attempts to activate $v$ are independent, and $v$ is
	activated as long as seeds from one strategy activates $v$.
This is a mixture of IC and LT models, and is our key to allow the reduction to work.

We denote the augmented graph together with the above described propagation model as $G_A$.
In $G_A$, only virtual nodes can be selected as seeds, and only real nodes are counted towards
	the influence spread.
The propagation in $G_A$ starts from the seeds in the virtual strategy nodes, and 
	these seeds activate real nodes according to the above IC and LT mixture model. 
Then the propagation among real nodes follow the original triggering model.
The reason this reduction works is justified by the following theorem.

\begin{restatable}{theorem}{thmVSNreduction}
\label{thm:VSNreduction}
Under the independent strategy activation model (Eq.~\eqref{eq:hv}), 
	(1) for any strategy mix $\vx=(x_1, \ldots, x_d) \in \cX$, 
	the distribution of the set of nodes
	activated by $\vx$ in the LIM model is the same as the distribution of the set of real nodes
	activated by seed set $S^{\vx} = \bigcup_{j=1}^d U_{j,x_j\delta^{-1}}$ in $G_A$.
(2) Conversely, for any seed set $S\subseteq U$, we can map $S$ to $\vx^S =(x^S_1, \ldots, x^S_d) $
	where $x^S_j = |S \cap U_j|\cdot \delta$, such that 
	the influence spread of $S$ in $G_A$
	(only counting the activation of the real nodes) is at most the influence spread
	of $\vx^S$ in the LIM model.
As a consequence, if an approximation algorithm for the triggering model 
	produces $S$ on graph $G_A$, then
	$\vx^S$ would be an approximate solution for LIM with the same approximation ratio.
\end{restatable}
\begin{proof}
	First, given strategy mix $\vx$, by the LT model and our weight construction 
	(Eq.~\eqref{eq:LTweight}), we know that the probability that the seed set $S^{\vx} \cap U_j = U_{j,x_j\delta^{-1}}$ 
	activates node $v$ in $G_A$ is $\sum_{i=1}^{x_j\delta^{-1}} w(u_{j,i},v) = q_{v,j}(x_j)$, which
	coincides with the probability that strategy $j$ with amount $x_j$ would activate $v$ in
	the LIM model.
	Among different strategy seed nodes, they attempt to activate $v$ independently, which
	coincide with Eq.~\eqref{eq:hv} that governs the activation of $v$ from strategy $\vx$.
	Since the remaining propagation among real nodes follows the same model, we can conclude that
	the set of nodes activated in either the LIM model or $G_A$ follows the same distribution. 
	
	Conversely, let $S$ be a seed set in $G_A$. 
	For each strategy $j$, $S \cap U_j$ may not be the prefix set.
	Let $U_{j,x^S_j}$ be the corresponding prefix set with $x^S_j = |S \cap U_j|$.
	We claim that $U_{j,x^S_j}$ activates $v$ with probability at least as high as
	that of $S \cap U_j$ activating $v$.
	Here, we need to critically use the concaveness of $q_{v,j}$: by its concaveness, we know that
	edge weight $w(u_{j,i},v)$ is non-increasing over $i$.
	Then the sum of weights of the prefix set $U_{j,x^S_j}$ to $v$ is at least as large as the 
	sum of the weights of $S \cap U_j$ to $v$.
	Thus, by the LT model, our claim holds.
	Once the claim holds, we know that by moving the seeds to the prefix we always have a higher probability
	of activating each real node. By the first part of the proof, we know that the prefix
	seed set exactly corresponds to the strategy mix $\vx^S = (x^S_1, \ldots, x^S_d)$.
	Therefore, the influence spread of $\vx^S$ in the LIM model
	must be at least as high as the influence spread of
	$S$ in $G_A$.
	
	The final part on the approximation algorithm becomes straightforward once we have the above
	results.
\hfill $\square$
\end{proof}

We remark that part (2) of the theorem critically depends on the concaveness of $q_{v,j}$, 
	and is where we need
	to use the LT model construction.
We could use the IC model with proper edge probability assignment for part (1), but
	it appears that IC model would not allow us to use the concaveness of $q_{v,j}$ to show
	part (2).
This is why we use a mixture of the IC and LT models in the end.

\begin{algorithm}[t]
	\caption{General Structure of Algorithm \IMMVSN}
	\begin{algorithmic} [1]
		\INPUT{$G$: the social graph; $\{D_v\}_{v\in V}$: triggering model parameters;
			$\{q_{v,j}\}_{v\in V, j\in S_v}$: strategy-node activation functions; 
			$k$: budget;
			$\delta$: granularity; $\varepsilon$: accuracy; $\ell$: confidence}
		\OUTPUT{$\vx \in \cX$}
		\STATE generate augmented graph $G_A$ and the diffusion model on it as follows: 
			(1) add virtual strategy nodes $U = \bigcup_{j=1}^d U_j$ to the node set, where 
				$U_j = \{u_{j,1}, u_{j,2}, \ldots, u_{j,k\delta^{-1}}\}$;
			(2) add directed edges $\{(u_{j,i},v) | v\in V, j \in S_v, u_{j,i}\in U_j\} $ to the edge set;
			(3) each edge $(u_{j,i},v)$ has LT weight $w(u_{j,i},v) = q_{v,j}(i\delta) -  q_{v,j}((i-1)\delta)$;
			(4) triggering set distribution of every real node $v$ is adjusted such that: 
				(4.1) real nodes are selected by $D_v$;
				(4.2) virtual nodes in $U_j$ with $j\in S_v$ are selected independent of real nodes and other virtual nodes; 
				(4.3) within $U_j$, virtual node $u_{j,i}$ is selected exclusively with probability
					$w(u_{j,i},v)$, just like in the LT model
		\STATE run $\IMM$ on graph $G_A$ with budget $k\delta^{-1}$ and obtain seed set $S\subseteq U$ on virtual nodes.
			$\IMM$ is adapted for $G_A$ as described in the text
		\STATE $\vx = (x^S_1, \ldots, x^S_d) $
		where $x^S_j = |S \cap U_j|\cdot \delta$
		\STATE \textbf{return} $\vx$
	\end{algorithmic}
	\label{alg:IMMVSN}
\end{algorithm}

With Theorem~\ref{thm:VSNreduction}, our algorithmic design for {\IMMVSN} is clear, and its general 
	structure is summarized in Algorithm~\ref{alg:IMMVSN}:
	We first construct the augmented graph $G_A$, and then apply an existing algorithm, in our case
	{\IMM}, on $G_A$ to find a seed set $S$ of virtual nodes with budget $k\delta^{-1}$, 
	and finally we convert $S$ to $\vx^S$ as specified in Theorem~\ref{thm:VSNreduction} as
	our solution.
When using {\IMM}, we also employ the following adaptations to improve its performance for
	the special $G_A$ graph:
(a) At each real node $v$ when we want to generate one more step in the reverse simulation, 
	we first sample $v$'s triggering set $T_v\sim D_v$ and put nodes in $T_v$ in the RR set, and these are real nodes;
	then for each strategy $j\in S_v$, we randomly pick at most one 
	virtual node $u_{j,i}$ with probability $w(u_{j,i},v)$ following the LT model, and 
	and this can be efficiently implemented by a binary search;
	finally, we do reverse simulation for each strategy $j$ independently, which corresponds to the
	independent activation across different strategies.
(b) Since only virtual nodes are seeds, an RR set without virtual nodes will be discarded, and
	greedy seed selection is only among the virtual nodes.
(c) Since only real nodes are counted towards the influence spread, we only uniformly at random pick
	roots of RR sets among real nodes.
(d) By part (2) of Theorem~\ref{thm:VSNreduction}, in the greedy {\NodeSelection}
	procedure of {\IMM} (corresponding to the {\LGreedy} procedure in {\IMMPRR}), 
	after selecting all the seed nodes, we 
	convert them to the prefix node set for each strategy.
(e) The total number of possible strategy mixes is at most $M= \min(k \delta^{-1} \ln d, d \ln(k\delta^{-1}))$ as discussed 
	in Section~\ref{sec:IMMPRR}, and together with
	part (d) above, we know the total number of seed set outputs is also at most 
	$M$, therefore, we will use $M$ to replace $\binom{n}{k}$ in
	the original {\IMM} algorithm.

The approximation guarantee of {\IMMVSN} is ensured by the correctness of the
	{\IMM} algorithm plus Theorem~\ref{thm:VSNreduction}.
For time complexity, our adaptions to {\IMM} save running time.
Overall, we have
\begin{restatable}{theorem}{thmIMMGMSVSN}
	Under the case of independent strategy activation (~Eq.\eqref{eq:hv}),
	the {\IMMVSN} algorithm
	returns a $(1-1/e-\varepsilon)$-approximate solution
	to the LIM problem with at least $1 - 1/n^\ell$ probability.
	When $q_{v,j}$'s are such that the
	optimal solution of LIM
	is at least as good as the best single node influence spread,
	{\IMMVSN} runs in $O((M + \ell \log n)(m + \log(k\delta^{-1})\sum_{v\in V} |S_v|) /\varepsilon^2)$ expected time, where $M= \min(k \delta^{-1} \ln d, d \ln(k\delta^{-1}))$.
	\label{thm:IMMGMSVSN}
\end{restatable}

The proof of the theorem follows that of~\cite{tang15}, and the novel part of the
	analysis is mainly summarized and proved in Theorem~\ref{thm:VSNreduction}.
The remaining part of the proof is given in Appendix~\ref{app:proofThm3}.
Comparing the running time result of Theorem~\ref{thm:IMMGMSVSN} with that of Theorem~\ref{thm:IMMGMS},
	we can see that the key difference is between the term $(m + \log(k\delta^{-1})\sum_{v\in V} |S_v|)$
	of {\IMMVSN} and the term $k\delta^{-1}\max_{v\in V}|S_v|(m+n)$ of {\IMMPRR}.
 {\IMMVSN} seems to have a better running time especially in avoiding an extra term of
 	$k\delta^{-1}$, which is partly because it does not require
 	maintaining partial RR sets, and partly because of the efficient LT reverse sampling method
 	via binary search.
Of course, these theoretical results are all upper bounds, so we cannot formally conclude
	the superiority of {\IMMVSN}.
We will demonstrate the superior performance of {\IMMVSN} through our empirical evaluation.
We also want to point out that {\IMMVSN} only works for the case of independent strategy activation,
	while {\IMMPRR} works for more general cases, and thus we cannot say that {\IMMVSN} can
	always replace {\IMMPRR}.

\vspace{-2mm}
\section{Experiments}
\vspace{-1mm}
\subsection{Experiment Setup}

\noindent
{\bf Datasets.\ \ }
We ran our experiments on 4 real-world networks, with statistics summarized in Table~\ref{table:stat}.
Three of them, denoted DM, NetHEPT, and DBLP, are collaboration networks:
	every node is an author and every edge means the two authors collaborated on a paper.
DM network is a network of data mining researchers
	extracted from the ArnetMiner archive (arnetminer.org) \cite{TangSWY09},
	NetHEPT is a network extracted from the high energy physics section of arxiv.org, while
	DBLP is extracted from the computer science bibliography database dblp.org~\cite{WCW12}.
Their sizes are small (679 nodes), medium (15K nodes), and large (654K) nodes, respectively.
We include the small DM dataset mainly to suit the slow Monte Carlo greedy algorithm.
The last dataset is Flixster, which is a user network of the movie rating site flixster.com.
Every node is a user and a directed edge from $u$ to $v$ means that $v$ has rated some movie(s)
that $u$ rated earlier \cite{barbieri2012topic}.
The IC model parameters of NetHEPT and DBLP are synthetically set using the weighted cascade
	method~\cite{kempe03journal}:
	edge $p(u,v) = 1/d_v$, where $d_v$ is the in-degree of node $v$.
For the DM and Flixster networks, we obtain learned edge parameters from the authors
	of~\cite{TangSWY09,barbieri2012topic} respectively.
\begin{table}[t]
	\centering
	\begin{tabular}{|c|c|c|c|}
	\hline
	Network & $n$ & $m$ & Average Degree \\ \hline \hline
	DM & 679 & 3,374 & 4.96  \\ \hline
	NetHEPT & 15,233 & 62,752 & 4.12 \\ \hline
	Flixster & 29,357 & 425,228  & 14.48  \\ \hline
	DBLP & 654,628 & 3,980,318 & 6.08  \\ \hline
	\end{tabular}
	\caption{Dataset Statistics}
	\label{table:stat}
\end{table}

\vspace{1mm}
\noindent
{\bf Application scenarios.\ \ }
We test two application scenarios of independent strategy activation explained in Section \ref{sec:application}.
The first is the {\em personalized marketing} scenario tested in~\cite{YangMPH16}.
In this scenario, each user $v$ has one unique strategy $x_v$ such as the personalized discount to
	$v$, $h_v(\vx)$ only depends on $x_v$.
We set $h_v(\vx) = 2 x_v - x_v^2$ following the same setting in~\cite{YangMPH16}.
The second one is the {\em segmented event marketing} scenario, which is not covered by previous studies.
In this case, each strategy $j$ is targeting at a disjoint subset of users $V_j$, and $x_j$ is the number
	of marketing events for user group $V_j$.
In our experiments, we set $d = 200$ for each dataset.
Moreover, we choose top $\min\{n, 2000\}$ nodes $V^*$ with the highest degrees
	from $V$.
For every node $v\in V^*$, we generate $i_v$ from $[d]$ uniformly at random
	and generate $r_{v, i_v}$ 
	from $[0, 0.3]$ uniformly at random.
For every $v\in V^*$, we set $S_v = \{i_v\}$ and $h_{v}(\vx) = 1 - (1 - r_{v, i_v})^{x_{i_v}}$;
for every  $v \in V\setminus V^*$, $S_v = \emptyset$.
This simulates the scenario where marketing efforts are focused on top connected nodes in the network.

\begin{figure*}[t]
	\centering
	\subfigure[DM] {
		\includegraphics[width=0.21\linewidth]{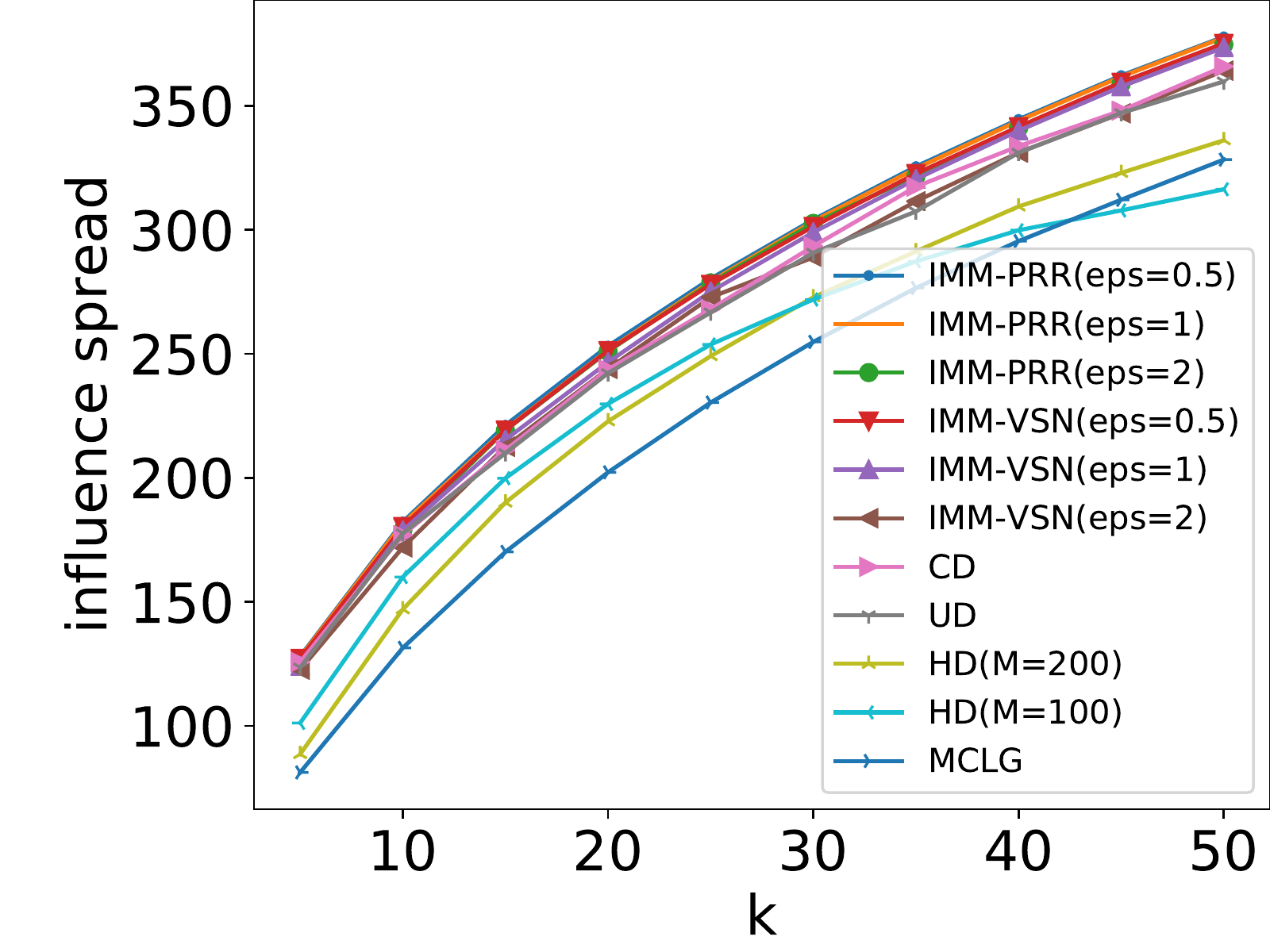}
	}
	\subfigure[NetHEPT] {
		\includegraphics[width=0.21\linewidth]{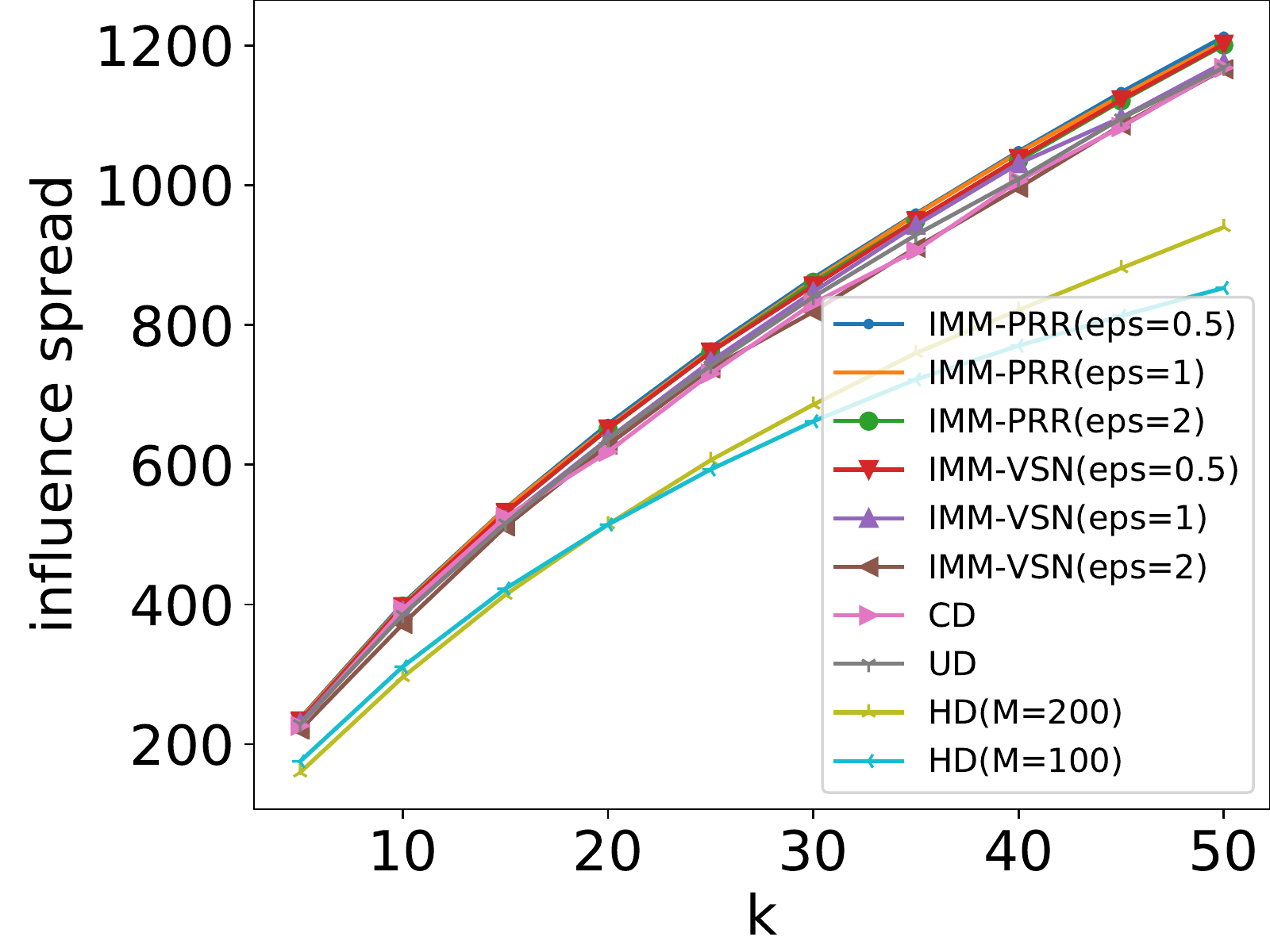}
	}
	\subfigure[Flixster] {
		\includegraphics[width=0.21\linewidth]{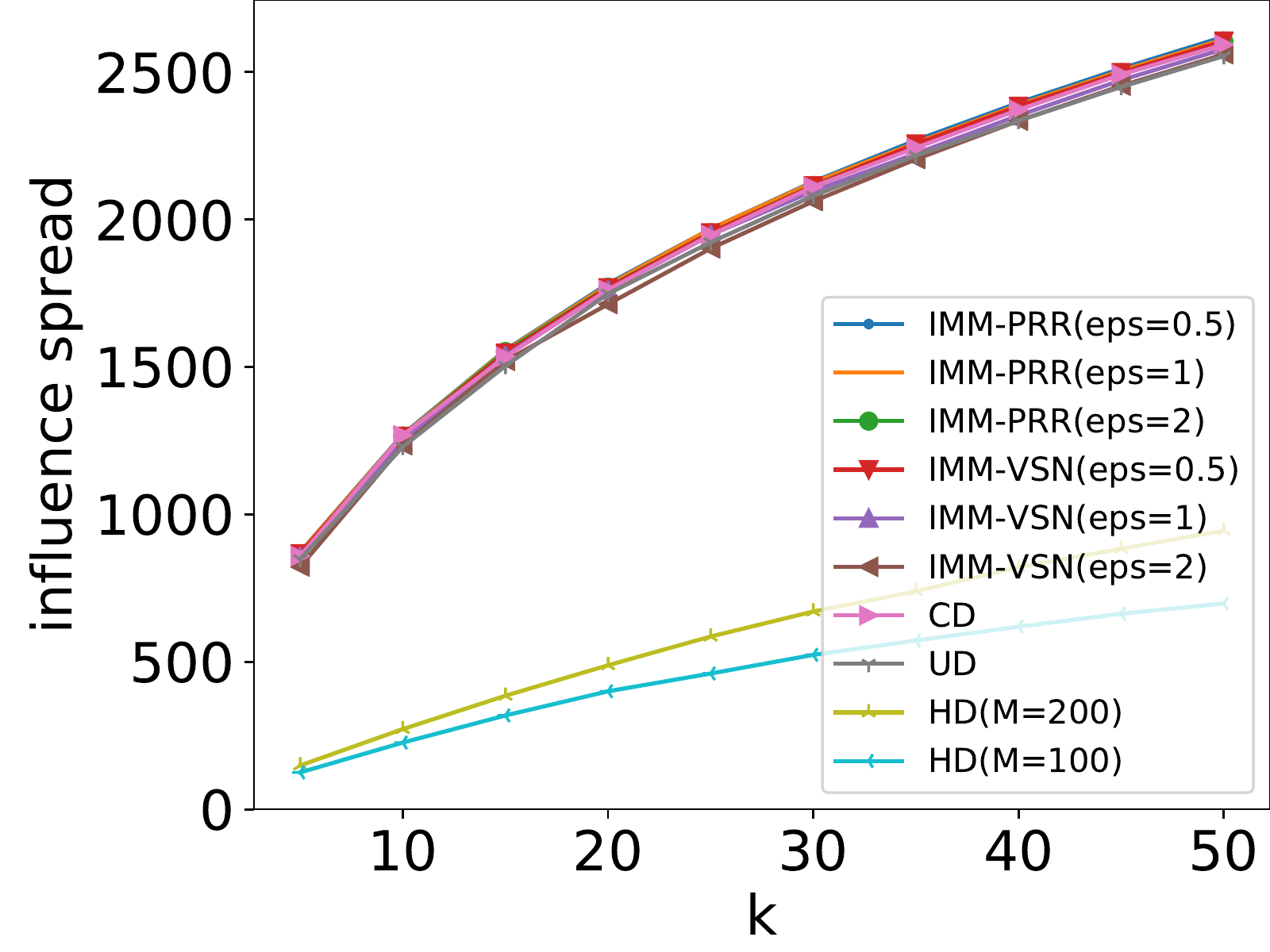}
	}
	\subfigure[DMLP] {
		\includegraphics[width=0.21\linewidth]{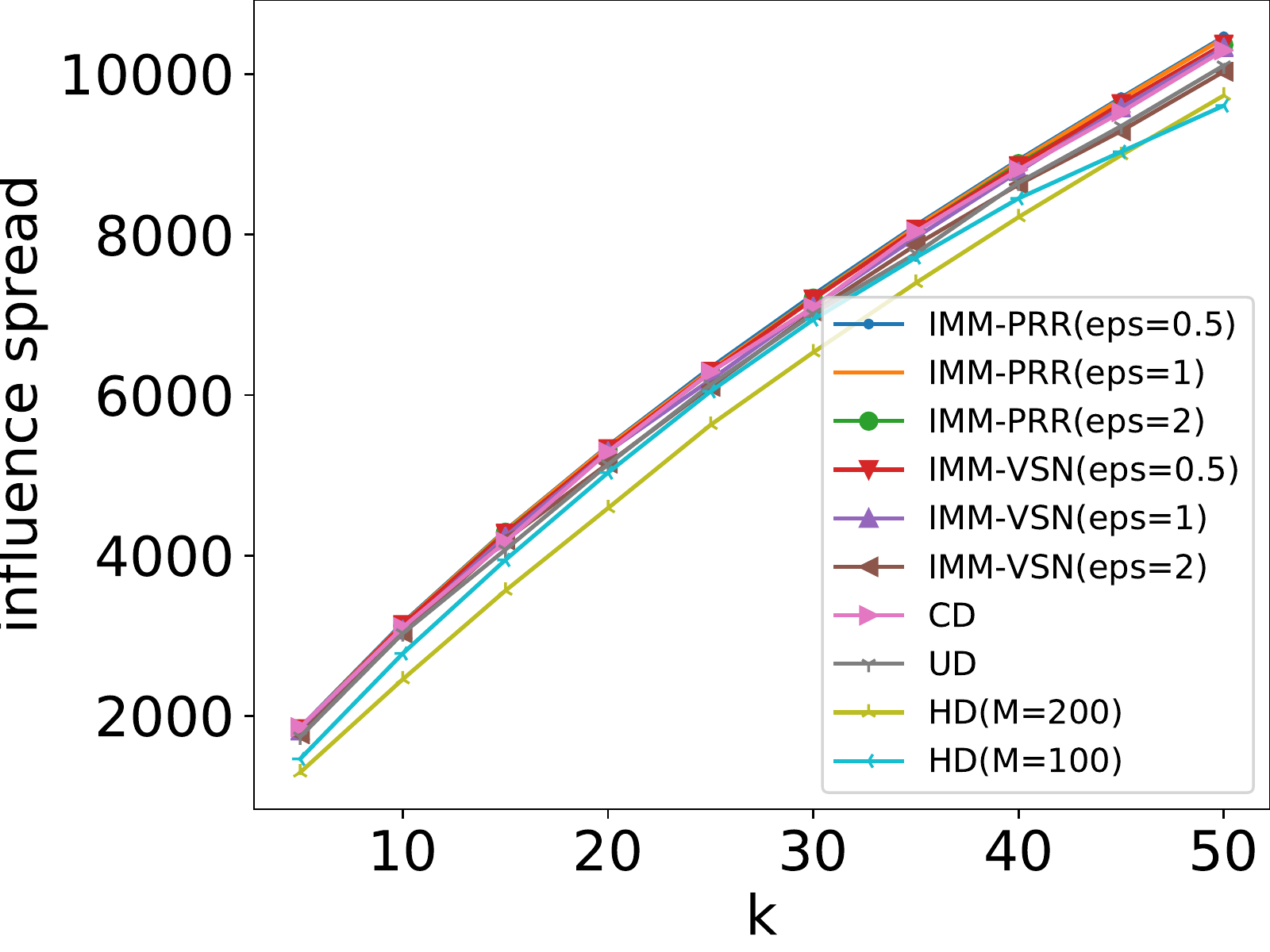}
	}
	\caption{Influence spread in personalized marketing scenario.}
	\label{fig:influence_spread_pm}
\end{figure*}
\begin{figure*}[h]
	\centering
	\subfigure[DM] {
		\includegraphics[width=0.21\linewidth]{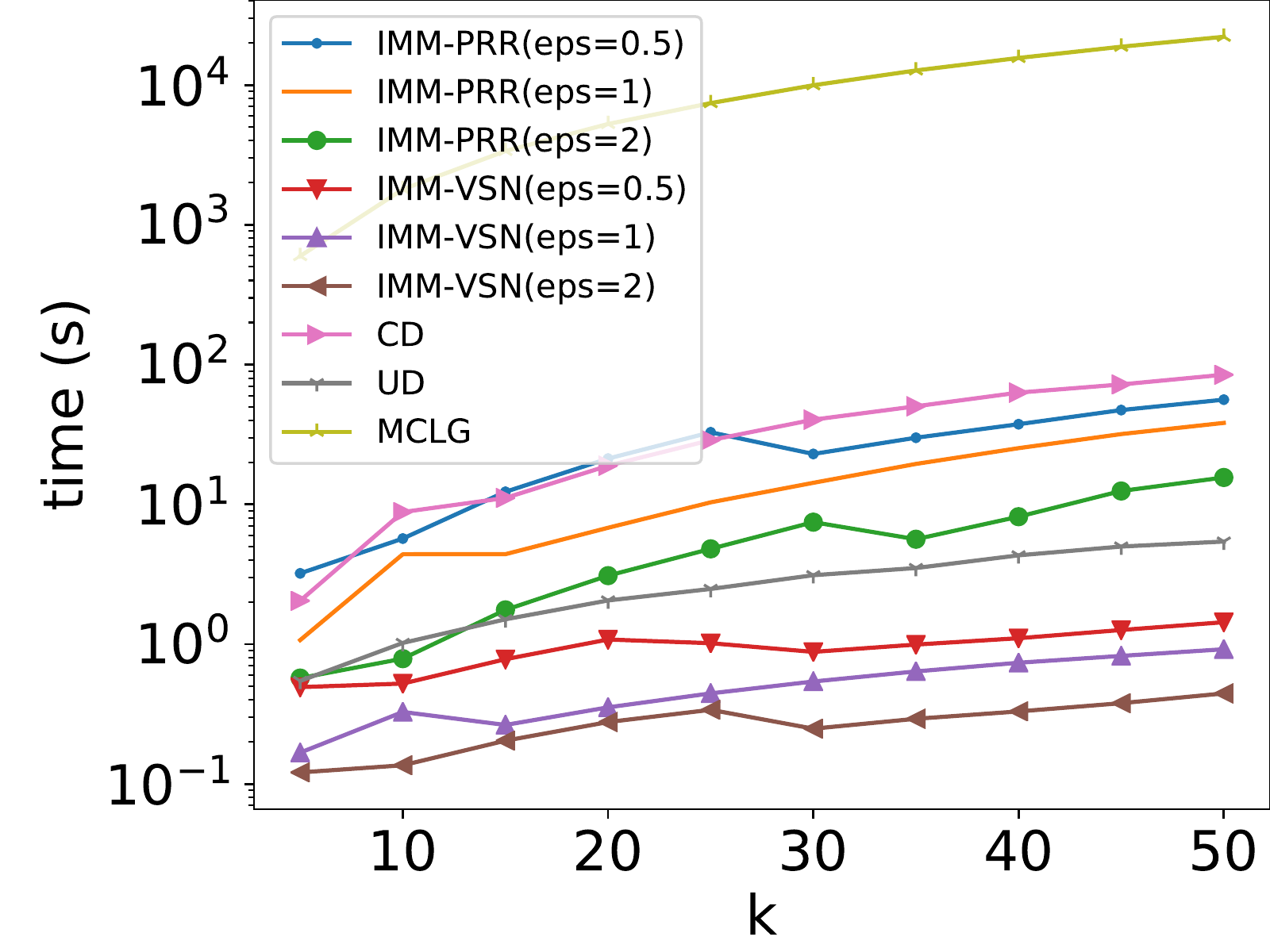}
	}
	\subfigure[NetHEPT] {
		\includegraphics[width=0.21\linewidth]{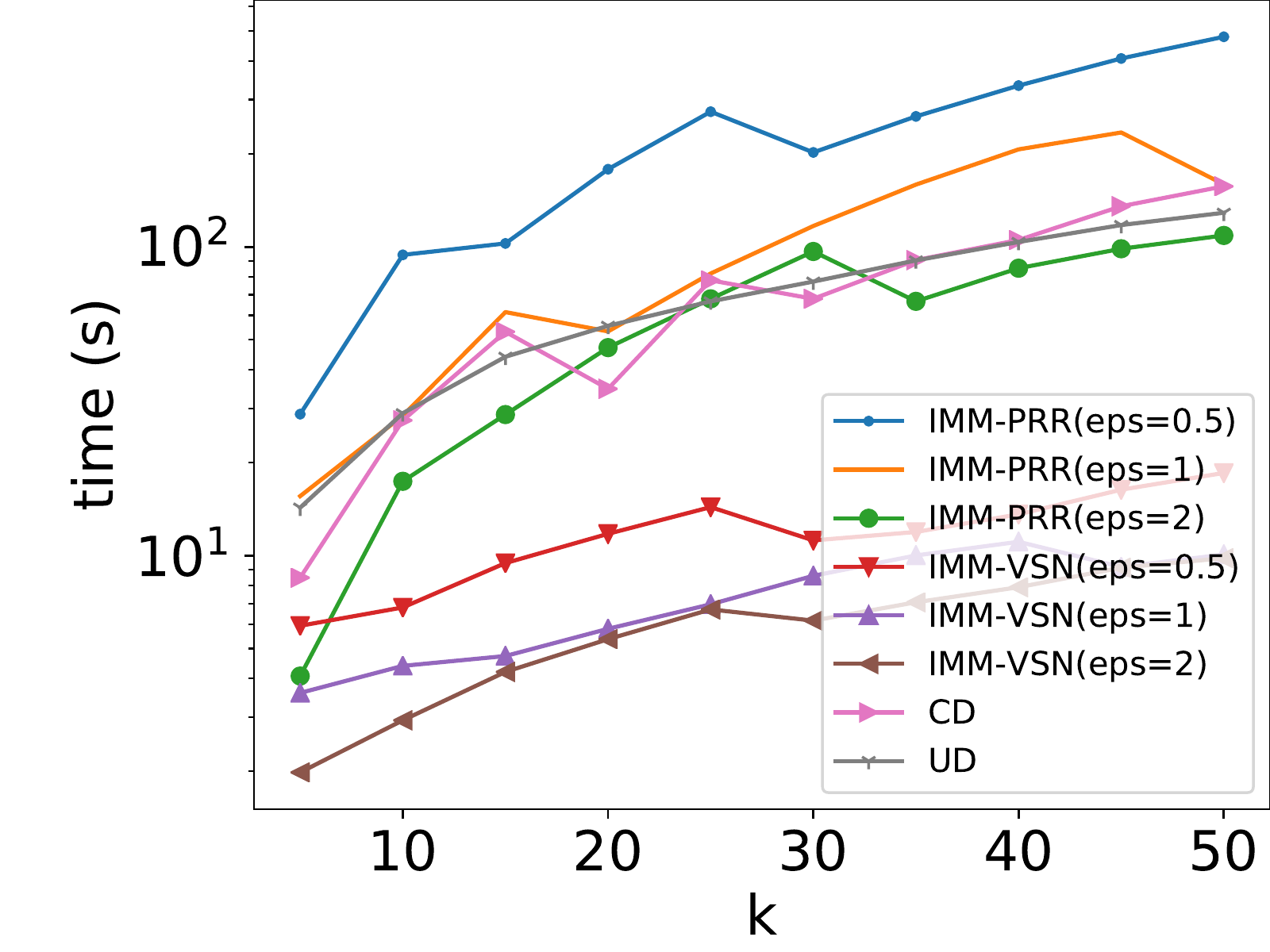}
	}
	\subfigure[Flixster] {
		\includegraphics[width=0.21\linewidth]{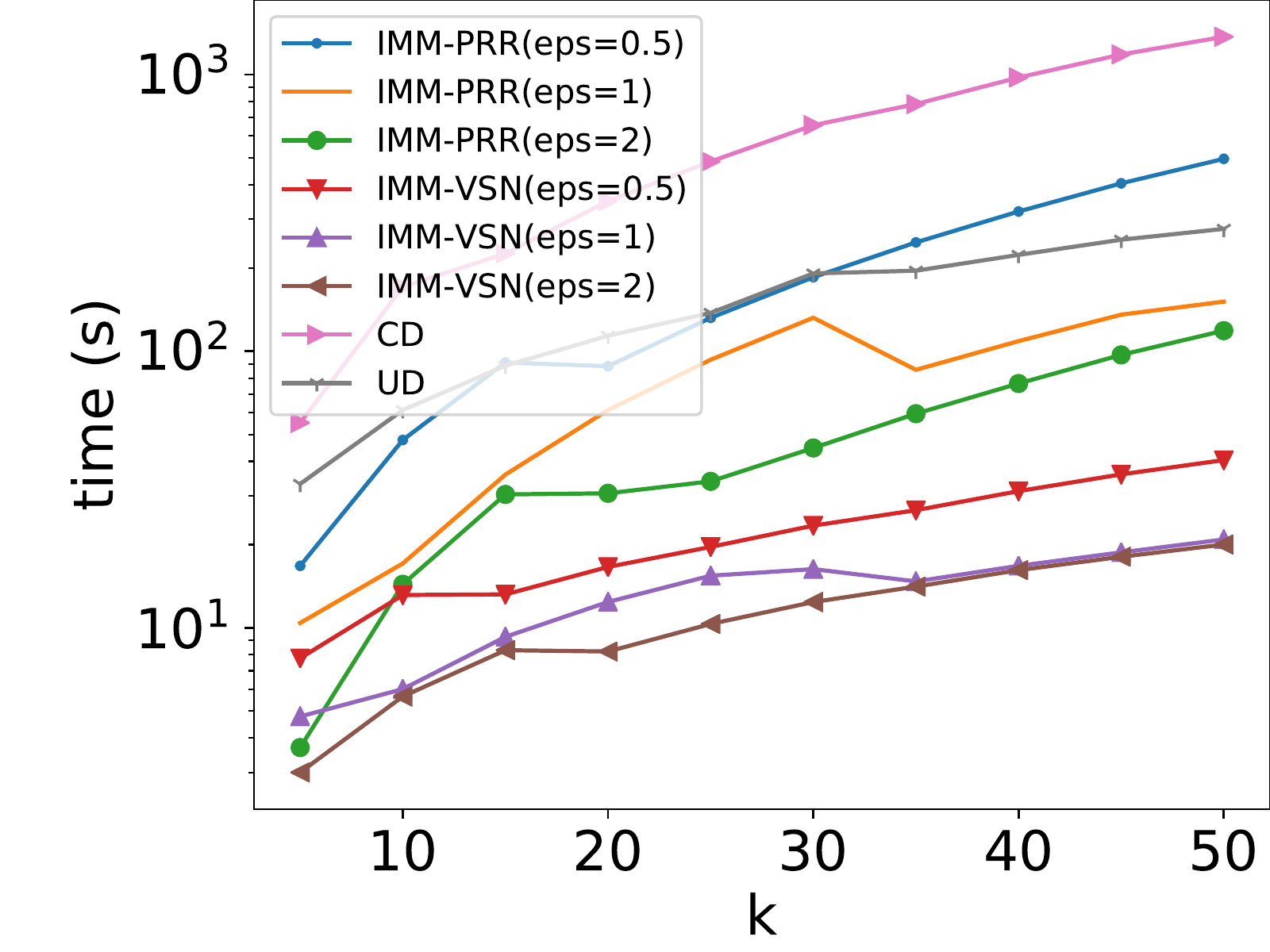}
	}
	\subfigure[DBLP] {
		\includegraphics[width=0.21\linewidth]{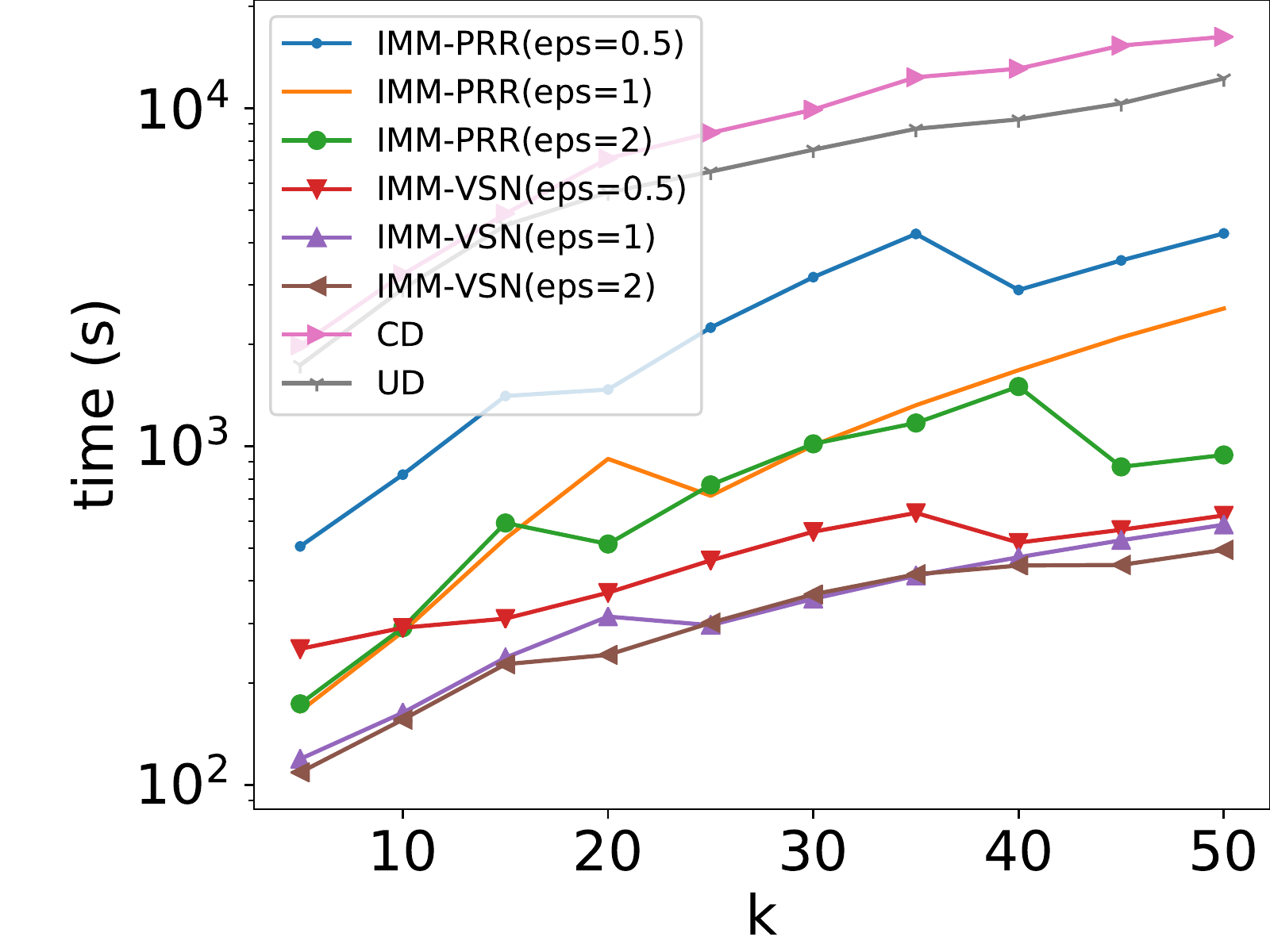}
	}
	\caption{Running time in the personalized marketing scenario.}
	\label{fig:time_pm}
\end{figure*}
\begin{figure*}[!h]
	\centering
	\subfigure[DM] {
		\includegraphics[width=0.21\linewidth]{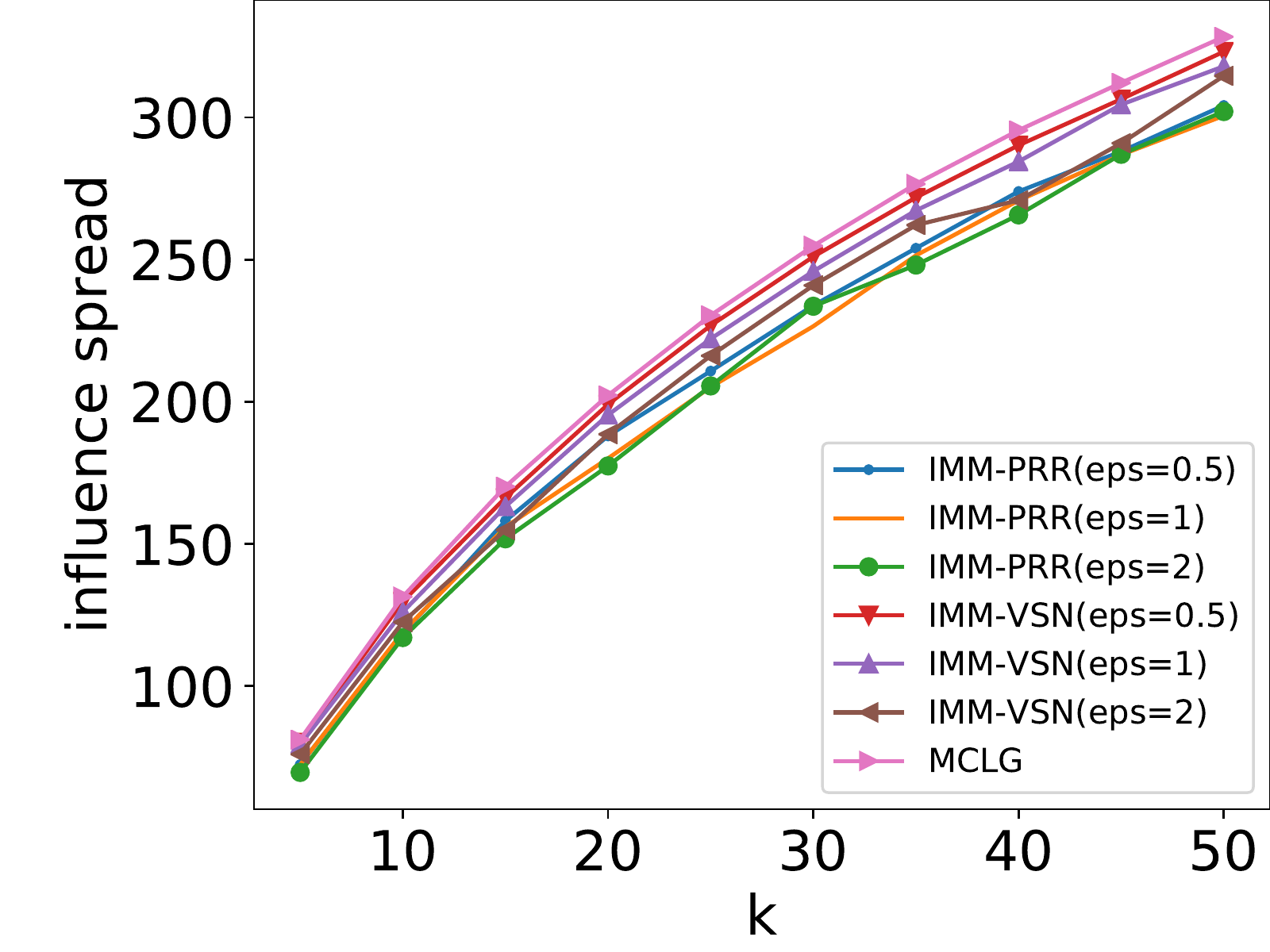}
	}
	\subfigure[NetHEPT] {
		\includegraphics[width=0.21\linewidth]{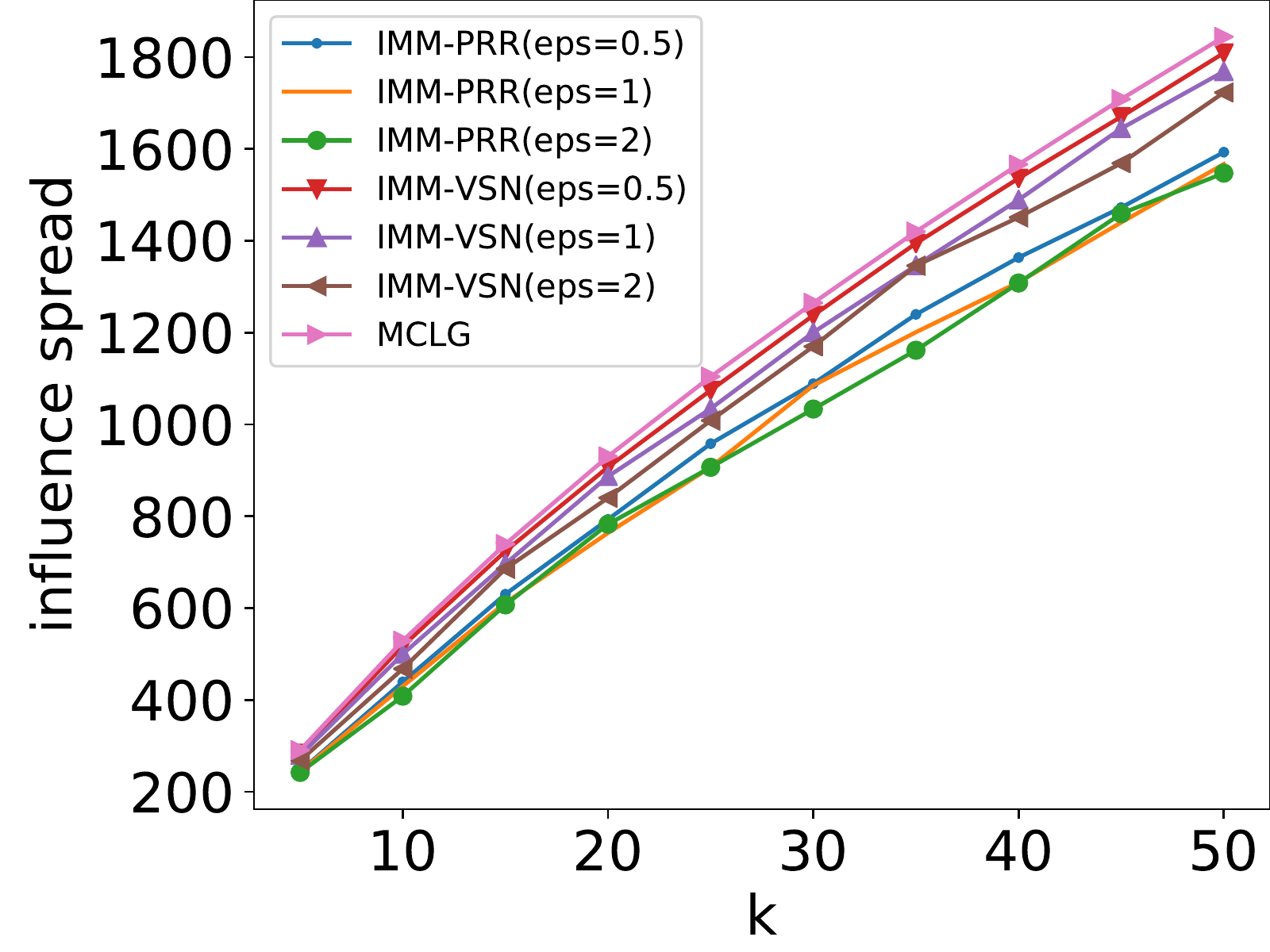}
	}
	\subfigure[Flixster] {
		\includegraphics[width=0.21\linewidth]{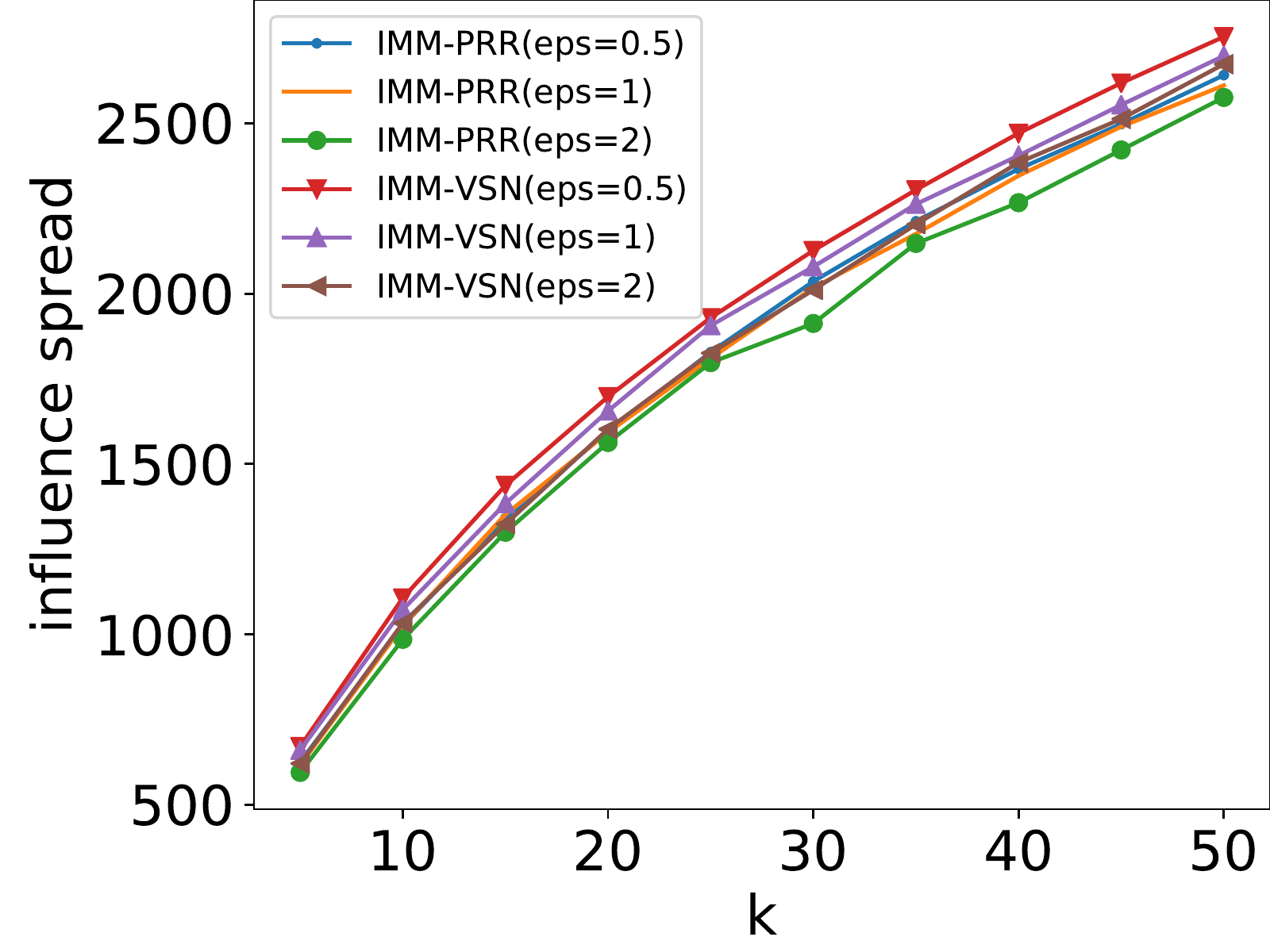}
	}
	\subfigure[DBLP] {
		\includegraphics[width=0.21\linewidth]{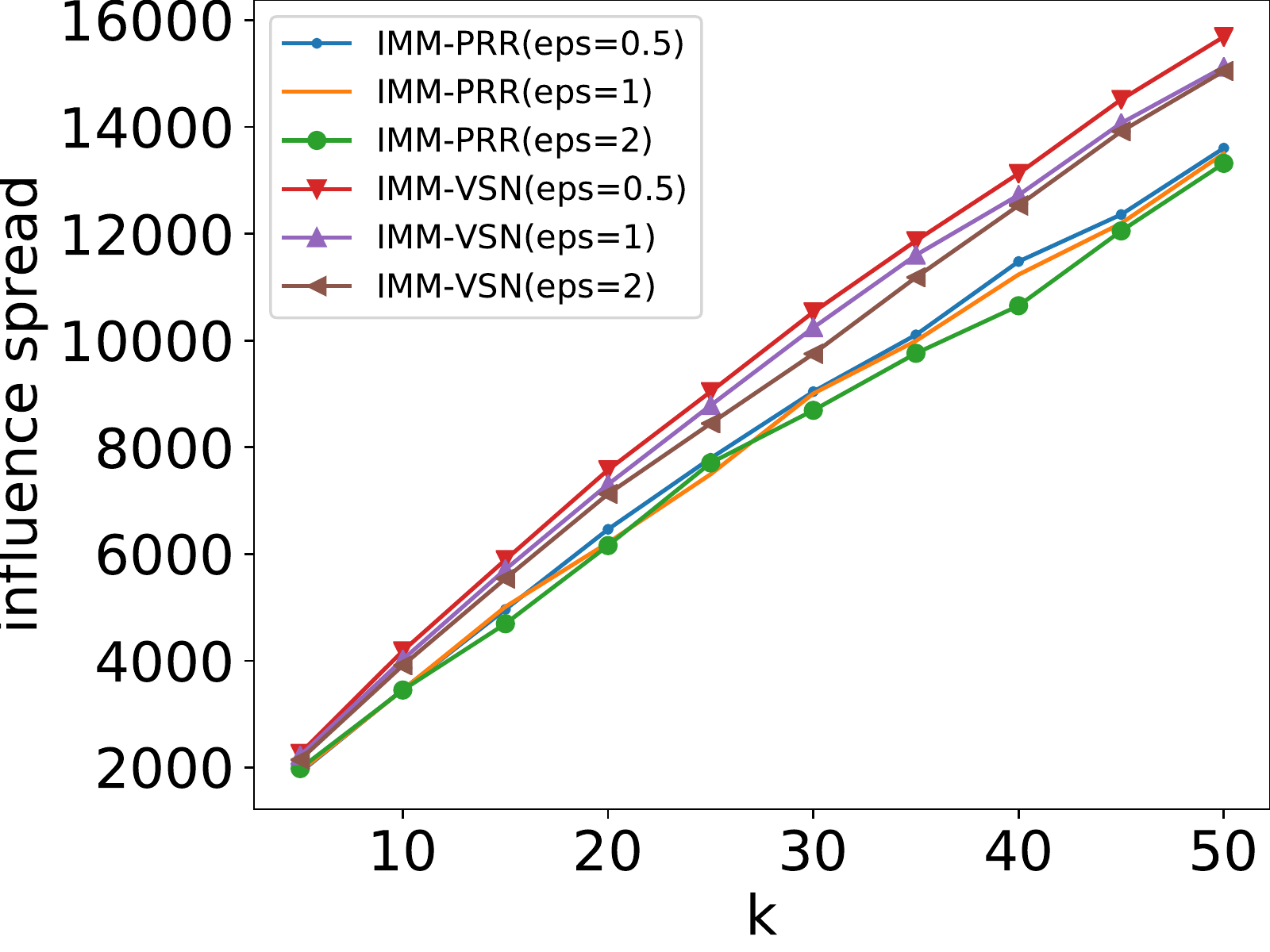}
	}
	\caption{Influence spread in the segmented event marketing scenario.}
	\label{fig:influence_spread_sem}
\end{figure*}
\begin{figure*}[!h]
	\centering
	\subfigure[DM] {
		\includegraphics[width=0.21\linewidth]{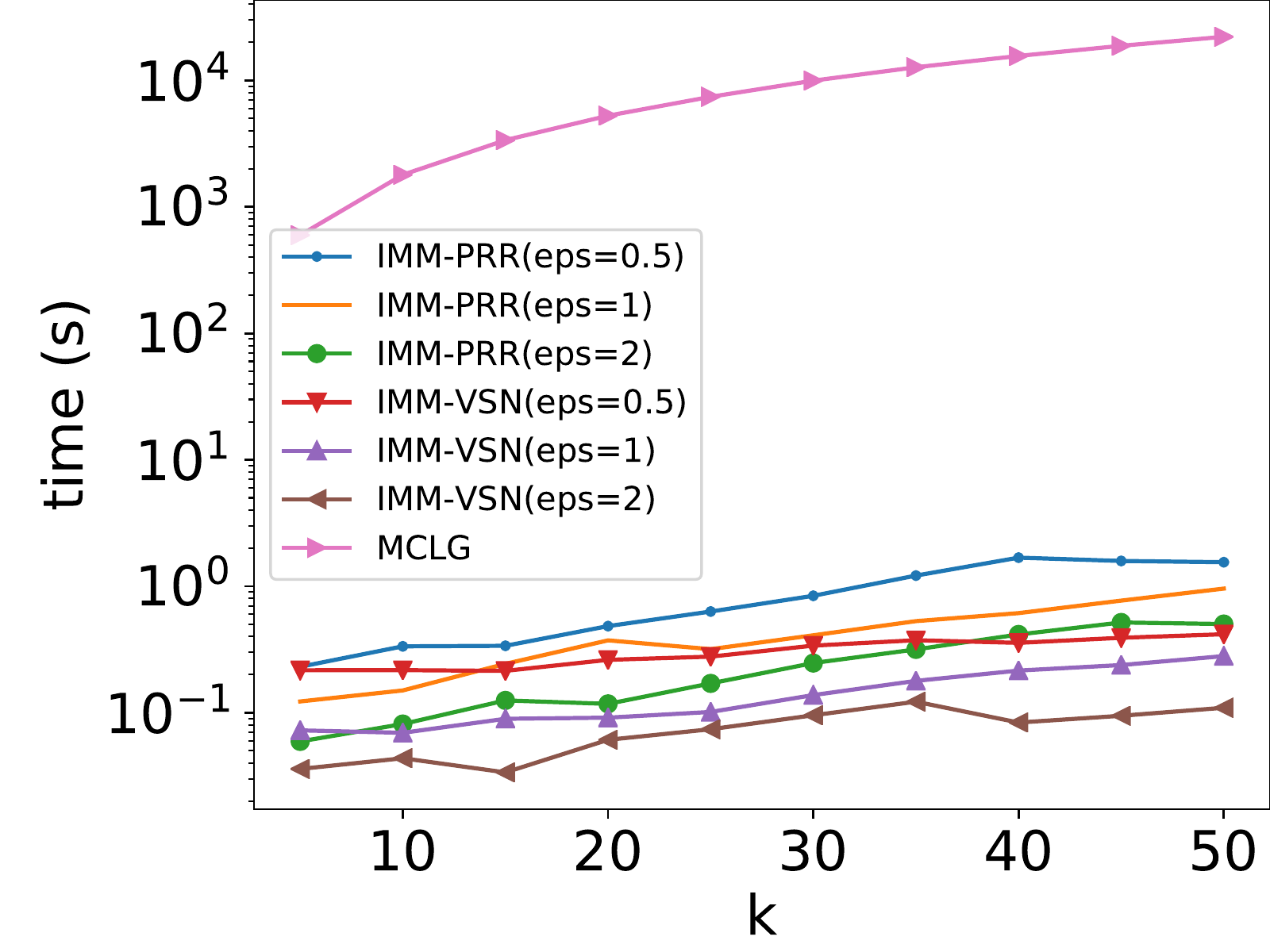}
	}
	\subfigure[NetHEPT] {
		\includegraphics[width=0.21\linewidth]{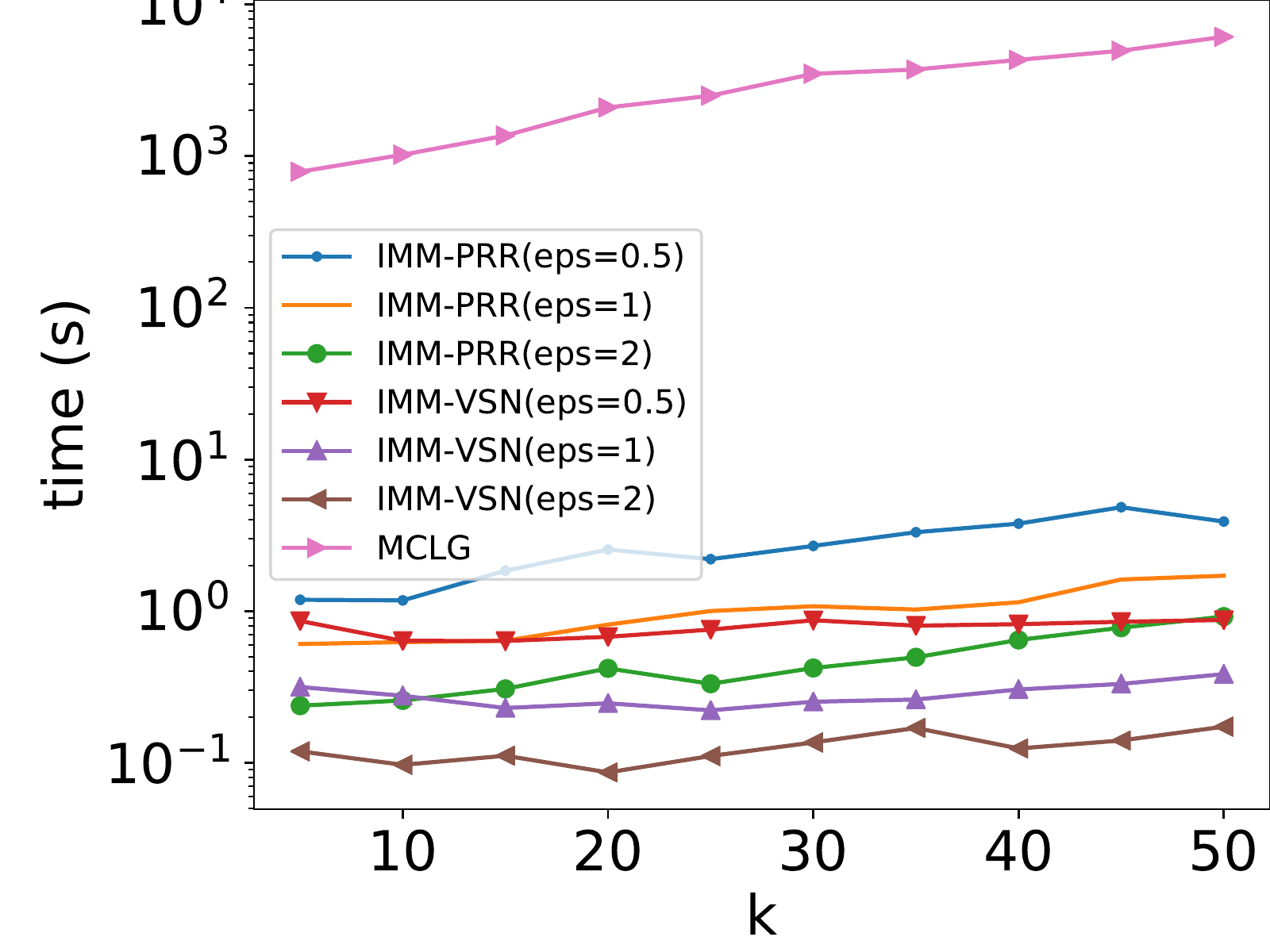}
	}
	\subfigure[Flixster] {
		\includegraphics[width=0.21\linewidth]{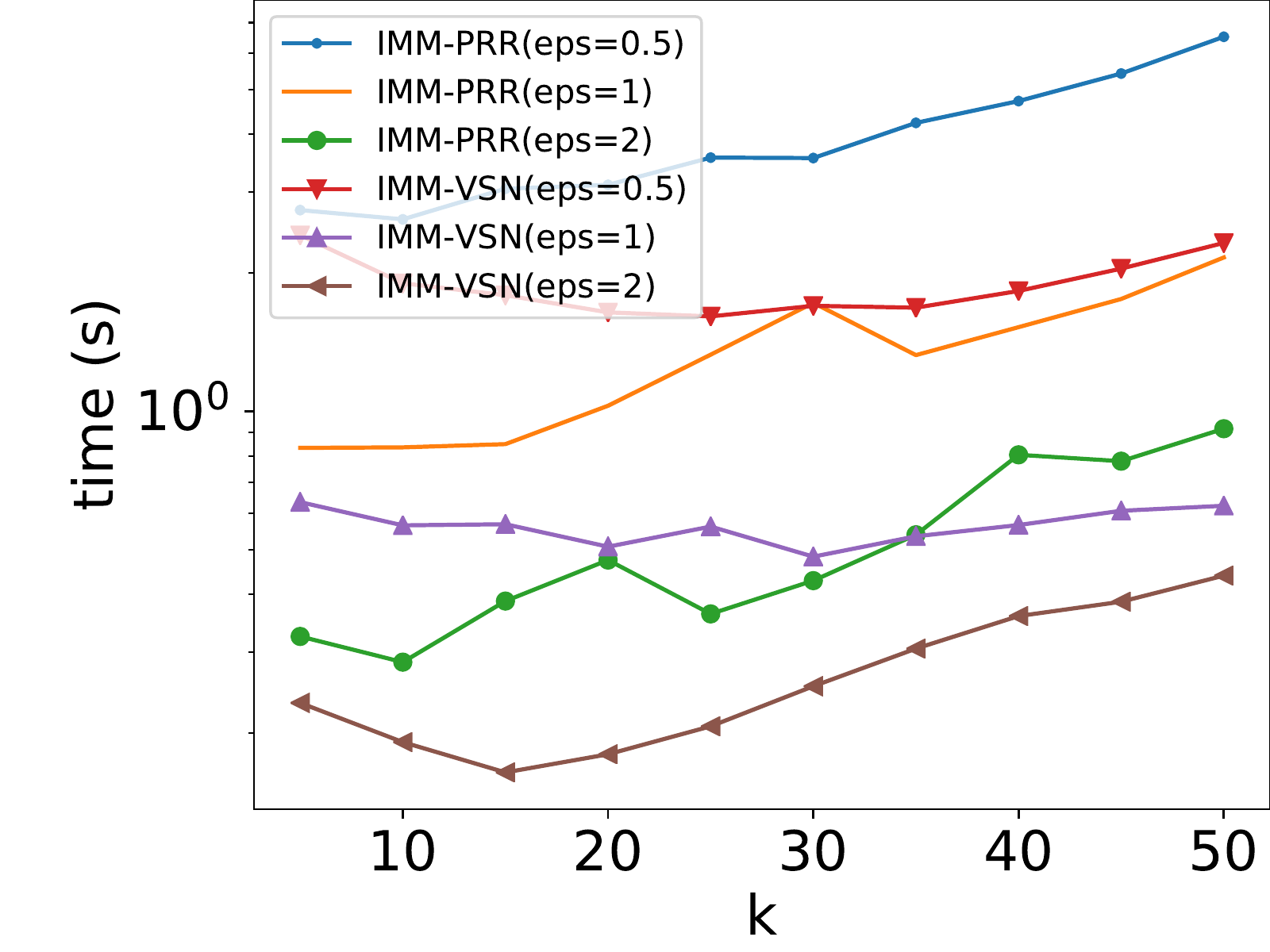}
	}
	\subfigure[DBLP] {
		\includegraphics[width=0.21\linewidth]{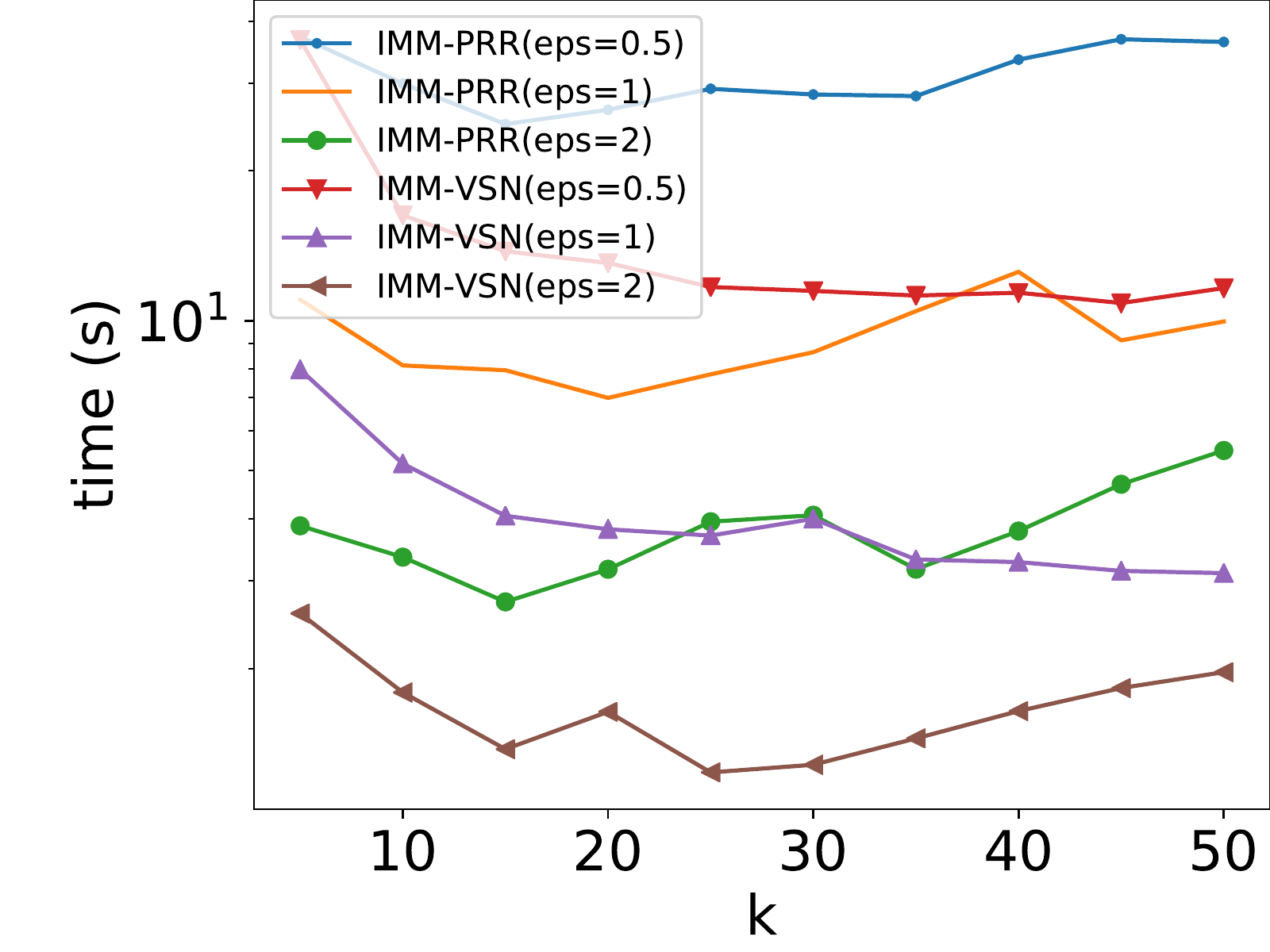}
	}
	\caption{Running time in the segmented event marketing scenario.}
	\label{fig:time_sem}
\end{figure*}

\vspace{1mm}
\noindent
{\bf Algorithms in Comparison.\ \ }
We test the following algorithms.
\begin{itemize}
\setlength{\itemsep}{0mm}
\item {\IMMPRR}/{\IMMVSN}. For both algorithms, we set $\ell = 1$, $\varepsilon = 0.5, 1, 2$.
When $\varepsilon = 1$ or $2$, {\IMMPRR}/{\IMMVSN} no longer has the approximation guarantee, but
	it is still a valid heuristic algorithm, since all other baselines are heuristic algorithms.

\item \UD. {\UD} is proposed in \cite{YangMPH16} for \textit{personalized marketing}. For each discount
	$c\in \{0.1, 0.2, \ldots, 1\}$, it will return a vector $\vx$ s.t. $x_i = 0$ or $x_i = c$ ($i\in [d]$). Then they run an exhaustive search of $c$ to find a best $c$.
\item {\CD}. {\CD} is also proposed in \cite{YangMPH16}. {\CD} uses the output of {\UD} as the initial
	value and runs a coordinate decent algorithm to achieve better result.

\item {\HD}. {\HD} is a heuristic baseline, where
	we choose top $M$ nodes with the highest degrees from $V$
	and then distribute the budget to those $M$ nodes proportional to their degrees.
	We set $M = 100$ and $M = 200$ in our experiments.
\item {\MCLG}. This is {\LGreedy} (Algorithm \ref{alg:hillclimb}) with Monte Carlo simulations to estimate
	influence spread $g(\vx)$. We use $10,000$ simulations for each estimation of $g(\vx)$.
\end{itemize}
For the personalized marketing scenario, we test all algorithms with granularity $\delta=0.1$.
For the segmented event marketing scenario, we do not test {\UD}, {\CD}, and {\HD}, since they are all
	designed for the personalized marketing scenarios.
In this case, $\delta=1$ as required by the scenario.
For all cases, we test total budget $k$ from $5$ to $50$.
We do not include the original influence maximization algorithm {\IMM} for seed set optimization
	in our tests, because \cite{YangMPH16} already demonstrates that
	the original {\IMM} is  inferior to {\UD} and {\CD} in influence spread.
%
%

All our tests are run on a Ubuntu 14.04.5 LTS server with 3.3GHz and 125GB memory. All algorithms are coded in C++ and
	compiled by g++.
All results on influence spread are the average of 10000 simulation runs for any given seed set, and all results on running time are the average of
	five algorithm runs.

\vspace{-2mm}
\subsection{Experimental Results}
We first look at the results for personalized marketing.
Figure~\ref{fig:influence_spread_pm} shows the influence spread result and Figure~\ref{fig:time_pm} shows the running
	time result.
First comparing between our two algorithms {\IMMPRR} and {\IMMVSN}, they produce about the same 
	influence spread but {\IMMVSN} typically runs much faster than {\IMMPRR}, in many cases close
	to or more than one-order of magnitude for the same parameter setting.
This demonstrates that the virtual strategy node approach indeed runs faster, matching our
	theoretical analysis.
Moreover, changing $\varepsilon$ from $0.5$ to $2$ significantly improves the running time with
	very slight or no penalty on influence spread.
	
When comparing to {\MCLG} algorithm (only run on DM),
	our {\IMMPRR}/{\IMMVSN} algorithms show clear advantage: its running time is two to four orders of magnitude faster than
	{\MCLG} while their influence spreads are also better than {\MCLG}.

When comparing to {\UD} and {\CD} heuristics,
	our {\IMMPRR}/{\IMMVSN} algorithms consistently perform
	better than {\UD} and {\CD} in influence spread.
For running time,  {\IMMVSN} runs much faster than {\UD} and {\CD} by one or two orders of magnitude,
	and {\IMMPRR} with $\varepsilon = 2$ is also faster than {\UD} and {\CD} (except on DM).
This again demonstrates the scalable design of our approach, in particular our algorithm
	{\IMMVSN} with $\varepsilon = 0.5$ can provide both theoretical guarantee and superior empirical performance in both influence spread and running time, 
	while neither {\UD} or {\CD} provides any theoretical guarantee.

For the baseline heuristic {\HD}, the result shows that its influence spread is significantly lower than others (especially in NetHEPT and Flixster), and thus it is
	not a competitive heuristic, even though it is very simple and fast.

The results on segmented event marketing are shown in Figures~\ref{fig:influence_spread_sem} and~\ref{fig:time_sem}.
{\MCLG} is too slow so is only run on the smaller DM and NetHEPT datasets.
Overall the results are consistent with the results for personalized marketing.
{\IMMVSN} typically runs much faster than {\IMMPRR}, and
	it runs 4-5 orders of magnitude faster than {\MCLG}.
Increasing $\varepsilon$ also significantly improve running time, with only slight decrease
	in influence spread.
In terms of influence spread, {\IMMVSN} with $\varepsilon = 0.5$ has the best influence spread among
	different settings for {\IMMPRR}/{\IMMVSN}, and 
	is only slightly lower than the influence spread achieved by {\MCLG}.

From these experiments, we can conclude that for the large class of independent strategy
	activation scenarios, {\IMMVSN} is the best choice that provides both theoretical guarantee
	and fast running time, and it outperforms the Monte Carlo greedy algorithm by several orders of
	magnitude, and is also significantly faster than other competing heuristic algorithms.
Moreover, our algorithms allow the easy tuning of parameter $\varepsilon$ to significantly improve
	running time with small or no penalty on influence spread.

\section{LIM with Partitioned Budgets}

In this section, we further generalize the LIM problem with partitioned budgets.
More specifically, marketing strategies often belong to multiple categories, and each category may be
assigned a separate budget.
Formally, the strategy set $[d]$ is partitioned into $\lambda$ categories $C_1, \ldots, C_\lambda$, and each
category $C_j$ has a budget $k_j$, i.e. $\sum_{i\in C_j} x_i \le k_j$.
For convenience, we use $\vx_C$ to denote the projection of vector $\vx$ into index set $C$.
Then the above constraint is $|\vx_{C_j}| \le k_j$. 
The partitioned budget problem is formally defined below.

\begin{definition}[Lattice Influence Maximization with Partitioned Budgets]
	\label{def:LIM-PB}
	Given the same input as in the LIM problem (Definition~\ref{def:LIM}), except that
	total budget $k$ is replaced by partitions $\{C_j\}_{j\in [\lambda]}$ and partitioned budgets
	$\{k_j\}_{j\in [\lambda]}$, 
	the task of lattice influence maximization with partitioned budgets, 
	denoted as LIM-PB, is to find an optimal strategy mix $\vx^*$ that achieves
	the largest influence spread within the partitioned budget constraints, that is
	\[
	\vx^* = \argmax_{\vx\in \cX, |\vx_{C_j}|\le k_j, \forall j\in [\lambda] } g(\vx).
	\]
\end{definition}

Note that since the per-strategy constraint $x_i \le b_i$ for the original LIM problem does not
	change the problem, our partitioned constraint here means that
	 $|C_j| > 1$ for all $j\in [d]$.
%
%

We next explain how to solve the partitioned budget constraint version LIM-PB.
Our method relies on the submodular maximization problem under the general matroid constraint.
A matroid on a set of elements $U$ is a collection of subsets of $U$ called independent sets, which satisfy
the following two properties: (a) If $I\subseteq U$ is an independent set, then every subset
of $I$ is also an independent set; and (b) If $I, I'$ are two independent sets with $|I| < |I'|$,
then there must be some element $ e\in I' \setminus I$ such that $I\cup \{e\}$ is also an
independent set.
The simplest matroid is the uniform matroid, where for some parameter $k$ all subsets $I$ with
$|I| \le k$ is an independent set.
Classical influence maximization essentially uses the uniform matroid constraint.
A partition matroid is such that, for a certain partition of $U$ into disjoint sets 
$A_1, \ldots, A_\lambda$,  and for parameters $k_1, \ldots, k_\lambda$, 
all subsets $I \subseteq U$ satisfying $ |I\cap A_i| \le k_i$ for all $i\in [\lambda]$ are
independent sets.
The classical result by \cite{fisher1978analysis} shows that the greedy algorithm
on a general matroid could achieve $1/2$ approximation ratio for nonnegative monotone and
submodular set functions.

Through Lemmas~\ref{lem:htogsubmodular} we already know that our objective functions $g(\vx)$
and $\eg_\cR(\vx)$ are nonnegative, monotone, and DR-submodular, but they are vector functions.
We now show how to translate them into equivalent set functions and then show that 
the LIM-PB problem corresponds to a partitioned matroid constraint under the set representation.
Let $b \ge \sum_{j\in[\lambda]}k_j \cdot \delta^{-1}$ be a large enough integer.
Construct the set of elements $U = \{(j, s)\mid j\in [d], s\in [b] \}$.
For any subset $A \subseteq U$, denote $A^{(j)} = A \cap \{(j,s)\mid s\in [b] \}$.
We map $A$ into a vector $\vx^A = (x^A_1, \ldots, x^A_d)$ such that $x^A_j = |A^{(j)}|\cdot \delta$.
Conversely, for every vector $\vx \in \cX$ satisfying the partitioned budget constraint,
we map $\vx$ to a set $A^{\vx} = \{ (j,s)\mid j\in [d], s \cdot \delta \le x_j  \}$.
For every vector function $f$, we define a set function $f^U$ on $U$ to be
$f^U(A) = f(\vx^A)$, for all $A \subseteq U$.
It is easy to see that the marginal $f^U(A\cup \{(j,s)\}) - f^U(A) = 
f(\vx^A + \delta\ve_j) - f(\vx^A)$.
Thus, one can verify that if $f$ is monotone and DR-submodular, then $f^U$ is monotone and submodular.
Next, for the partitioned budget constraint  $|\vx_{C_i}| \le k_i$ given partition
$C_1, \ldots, C_\lambda$ of $[d]$ and budgets $k_1, \ldots, k_\lambda$,
it is equivalent to partition $U$ to $U_1, \ldots, U_\lambda$, with $U_i = C_i \times [b]$,
and enforce constraint $|A \cap U_i| \le k_i \cdot \delta^{-1}$ for all $A\subseteq U$ and
$i \in [\lambda]$. 
Therefore, we translate the LIM-PB problem of maximizing $g(\vx)$ with the partitioned budget
constraint to maximizing $g^U(A)$ under the partition matroid constraint.
Similarly we can translate $\eg(\vx)$ to $\eg^U(A)$.
Therefore, we can conclude that the greedy algorithms under the partitioned budget
constraint could achieve $1/2$ approximation.

The actual greedy algorithm is straightforward.
In $\IMMPRR$, in every greedy step when we need to find
	another increment in one of the strategies  
	(line~\ref{line:argmaxgreedy} of Algorithm~\ref{alg:hillclimb} or 
	line~\ref{line:largestj} of Algorithm~\ref{alg:hill_climbing_delta}), 
	instead of taking $\argmax$ among all possible $j\in [d]$, we only search for $j$
	such that $\vx + \delta \ve_j$ still satisfies the partitioned budget constraint.
Similarly, in $\IMMVSN$, when we need to find another virtual strategy node as a seed, 
	we need to only search for those seeds that would satisfy the partitioned budget constraint.
The greedy steps terminates until the partitioned budgets are exhausted.
The corresponding algorithms achieves $1/2-\varepsilon$ approximation ratio with probability at
	least $1 - 1/n^\ell$, and runs in the same expected running time as in their non-partitioned versions.

\vspace{-1mm}
\section{Conclusion and Future Work}

We design two RIS-based scalable algorithms, {\IMMPRR} based on partial RR sets and {\IMMVSN} based on virtual
	strategy nodes, that guarantee $1-1/e-\varepsilon$ approximation to the lattice influence maximization
	problem.
{\IMMPRR} could solve the general LIM problem, while {\IMMVSN} has better running time for the case
	of independent strategy activations, as demonstrated both empirically and through theoretical
	analysis.

There are several future directions to this study.
One direction is to study continuous domain, and investigate if RIS-based approach can be adapted
	to the continuous domain.
Another direction is to study how to apply gradient methods for continuous influence maximization.
It may also be interesting to study lattice or continuous influence maximization
	in other influence propagation settings
	such as competitive influence maximization.

\bibliographystyle{splncs03}
\bibliography{bibdatabase}

\begin{thebibliography}{10}
\providecommand{\url}[1]{\texttt{#1}}
\providecommand{\urlprefix}{URL }

\bibitem{barbieri2012topic}
Barbieri, N., Bonchi, F., Manco, G.: Topic-aware social influence propagation
  models. In: ICDM. pp. 81--90. IEEE (2012)

\bibitem{BorgsBrautbarChayesLucier}
Borgs, C., Brautbar, M., Chayes, J., Lucier, B.: Maximizing social influence in
  nearly optimal time. In: SODA. pp. 946--957 (2014)

\bibitem{BAA11}
Budak, C., Agrawal, D., Abbadi, A.E.: {Limiting the spread of misinformation in
  social networks}. In: WWW. pp. 665--674 (2011)

\bibitem{ChenFLFTT15}
Chen, S., Fan, J., Li, G., Feng, J., Tan, K., Tang, J.: Online topic-aware
  influence maximization. {PVLDB}  8(6),  666--677 (2015)

\bibitem{Chen18}
Chen, W.: An issue in the martingale analysis of the influence maximization
  algorithm imm. In: CSoNet (2019)

\bibitem{chen2013information}
Chen, W., Lakshmanan, L.V., Castillo, C.: Information and Influence Propagation
  in Social Networks. Morgan \& Claypool Publishers (2013)

\bibitem{ChenWY09}
Chen, W., Wang, Y., Yang, S.: Efficient influence maximization in social
  networks. In: KDD. pp. 199--208 (2009)

\bibitem{CYZ10}
Chen, W., Yuan, Y., Zhang, L.: Scalable influence maximization in social
  networks under the linear threshold model. In: ICDM. pp. 88--97 (2010)

\bibitem{CDPW14}
Cohen, E., Delling, D., Pajor, T., Werneck, R.F.: Sketch-based influence
  maximization and computation: Scaling up with guarantees. In: CIKM. pp.
  629--638 (2014)

\bibitem{DemaineHMMRSZ14}
Demaine, E.D., Hajiaghayi, M., Mahini, H., Malec, D.L., Raghavan, S., Sawant,
  A., Zadimoghaddam, M.: How to influence people with partial incentives. In:
  WWW (2014)

\bibitem{domingos01}
Domingos, P., Richardson, M.: Mining the network value of customers. In: KDD.
  pp. 57--66 (2001)

\bibitem{FeldmanNS11}
Feldman, M., Naor, J., Schwartz, R.: A unified continuous greedy algorithm for
  submodular maximization. In: FOCS. pp. 570--579 (2011)

\bibitem{fisher1978analysis}
Fisher, M.L., Nemhauser, G.L., Wolsey, L.A.: An analysis of approximations for
  maximizing submodular set functions—ii. In: Mathematical Programming Study,
  vol.~8, pp. 73--87 (1978)

\bibitem{GomezRodriguez16}
Gomez{-}Rodriguez, M., Song, L., Du, N., Zha, H., Sch{\"{o}}lkopf, B.:
  Influence estimation and maximization in continuous-time diffusion networks.
  {ACM} Trans. Inf. Syst.  34(2),  9:1--9:33 (2016)

\bibitem{simpath11}
Goyal, A., Lu, W., Lakshmanan, L.V.S.: {SIMPATH: An Efficient Algorithm for
  Influence Maximization under the Linear Threshold Model}. In: ICDM. pp.
  211--220 (2011)

\bibitem{HassaniSK17}
Hassani, S.H., Soltanolkotabi, M., Karbasi, A.: Gradient methods for submodular
  maximization. In: NIPS. pp. 5843--5853 (2017)

\bibitem{HeSCJ12}
He, X., Song, G., Chen, W., Jiang, Q.: {Influence Blocking Maximization in
  Social Networks under the Competitive Linear Threshold Model}. In: SDM. pp.
  463--474 (2012)

\bibitem{JungHC12}
Jung, K., Heo, W., Chen, W.: {IRIE: Scalable and Robust Influence Maximization
  in Social Networks}. In: ICDM. pp. 918--923 (2012)

\bibitem{kempe03}
Kempe, D., Kleinberg, J.M., Tardos, {\'E}.: Maximizing the spread of influence
  through a social network. In: Proceedings of the 9th ACM SIGKDD International
  Conference on Knowledge Discovery and Data Mining (KDD). pp. 137--146 (2003)

\bibitem{kempe03journal}
Kempe, D., Kleinberg, J.M., Tardos, {\'E}.: Maximizing the spread of influence
  through a social network. Theory of Computing  11(4),  105--147 (2015),
  conference version appeared in KDD'2003

\bibitem{Leskovec07}
Leskovec, J., Krause, A., Guestrin, C., Faloutsos, C., VanBriesen, J.M.,
  Glance, N.S.: Cost-effective outbreak detection in networks. In: KDD. pp.
  420--429 (2007)

\bibitem{lu2015competition}
Lu, W., Chen, W., Lakshmanan, L.V.: From competition to complementarity:
  comparative influence diffusion and maximization. PVLDB  9(2),  60--71 (2015)

\bibitem{NWF78}
Nemhauser, G.L., Wolsey, L.A., Fisher, M.L.: An analysis of the approximations
  for maximizing submodular set functions. Mathematical Programming  14(1),
  265--294 (1978)

\bibitem{NguyenTD16}
Nguyen, H.T., Thai, M.T., Dinh, T.N.: Stop-and-stare: Optimal sampling
  algorithms for viral marketing in billion-scale networks. In: SIGMOD. pp.
  695--710 (2016)

\bibitem{richardson02}
Richardson, M., Domingos, P.: Mining knowledge-sharing sites for viral
  marketing. In: KDD. pp. 61--70 (2002)

\bibitem{ShakarianSPB14}
Shakarian, P., Salmento, J., Pulleyblank, W.R., Bertetto, J.: Reducing gang
  violence through network influence based targeting of social programs. In:
  KDD. pp. 1829--1836 (2014)

\bibitem{SomaKIK14}
Soma, T., Kakimura, N., Inaba, K., Kawarabayashi, K.: Optimal budget
  allocation: Theoretical guarantee and efficient algorithm. In: ICML. pp.
  351--359 (2014)

\bibitem{SomaY15}
Soma, T., Yoshida, Y.: A generalization of submodular cover via the diminishing
  return property on the integer lattice. In: NIPS. pp. 847--855 (2015)

\bibitem{TangSWY09}
Tang, J., Sun, J., Wang, C., Yang, Z.: Social influence analysis in large-scale
  networks. In: KDD (2009)

\bibitem{TangTXY18}
Tang, J., Tang, X., Xiao, X., Yuan, J.: Online processing algorithms for
  influence maximization. In: SIGMOD. pp. 991--1005 (2018)

\bibitem{tang15}
Tang, Y., Shi, Y., Xiao, X.: Influence maximization in near-linear time: a
  martingale approach. In: SIGMOD. pp. 1539--1554 (2015)

\bibitem{tang14}
Tang, Y., Xiao, X., Shi, Y.: Influence maximization: near-optimal time
  complexity meets practical efficiency. In: SIGMOD (2014)

\bibitem{WCW12}
Wang, C., Chen, W., Wang, Y.: {Scalable influence maximization for independent
  cascade model in large-scale social networks}. Data Mining and Knowledge
  Discovery  25(3),  545--576 (2012)

\bibitem{YangMPH16}
Yang, Y., Mao, X., Pei, J., He, X.: Continuous influence maximization: {W}hat
  discounts should we offer to social network users? In: SIGMOD. pp. 727--741
  (2016)

\end{thebibliography}

\clearpage

\appendix

\section*{Appendix}

\section{Remaining Part of the Proof of Theorem~\ref{thm:IMMGMS}} \label{app:proofThm1}

%

The remaining part of the proof of Theorem~\ref{thm:IMMGMS} is directly modified from the proof of Theorem 4 in \cite{tang15} together
	with the fix in~\cite{Chen18}.
\begin{lemma}
\label{lem:hill}
Given $k, \delta, d, n, \ell$, Algorithm~\ref{alg:hill_climbing_delta} returns  a $(1 - 1/e - \varepsilon)$-approximation with at least $1 - 1 / n^\ell$ probability if $\theta$, the size of $\cR$, is at least $\lambda^*(\ell)/ \OPT$, where $\lambda^*(\ell)$ is defined in Eq.\eqref{eq:lambda*}. 
\end{lemma}
\begin{proof}
Denote $\vx^*$ as the solution of Algorithm ~\ref{alg:hill_climbing_delta} and $\vx^{\circ}$ as the optimal solution of  LIM problem. Through replacing the number of possible $k$-seed set $\binom{n}{k}$ of Lemma 3 and 4 in \cite{tang15} by the number of possible allocations $M$ in our problem, we can derive that with $1-n^\ell$ probability,
$$
\eg_{\cR}(\vx^{\circ}) \geq \left(1 - \varepsilon \cdot \frac{\alpha}{(1 - 1 / e)\cdot\alpha + \beta}\right) \cdot \OPT
$$
and
$$
\eg_{\cR}(\vx^*) \leq g(\vx^*) + \left(\varepsilon - \frac{(1 - 1 / e)\varepsilon \alpha}{(1 - 1 / e)\alpha + \beta}\right)\cdot \OPT
$$
Then by combining the greedy property that $\eg_{\cR}(\vx^*)\geq (1 - 1/e)\eg_{\cR}(\vx^{\circ})$, we have,
$g(\vx^*)\geq (1 - 1/e - \varepsilon)\cdot \OPT$.
\hfill $\square$
\end{proof}

\begin{lemma}
\label{lem:sample}
Let $\ell$ be the input of Algorithm~\ref{alg:sampling}.
With at least $1 - 1 /2n^{(\ell+\gamma)}$ probability, Algorithm \ref{alg:sampling} returns a set $\cR$ of RR sets with $|\cR|\geq \lambda^*(\ell+\gamma) / \OPT$, where $\lambda^*(\ell)$ is as defined in Eq.\eqref{eq:lambda*} and $\gamma$ is obtained in line~\ref{alg:workaroundb}.
\end{lemma}
\begin{proof}
Through replacing the number of possible $k$-seed set $\binom{n}{k}$ of Lemma 6 and 7 in \cite{tang15} by the number of possible allocations $M$ in our problem, we can easily get the result of this lemma. It's $1 - 1 /2n^{(\ell+\gamma)}$ rather than $1 - 1 /n^\ell$ because we reset
	$\ell$ as $\ell=\ell+\gamma + \ln 2 / \ln n$ in Algorithm \ref{alg:sampling}, line~\ref{alg:workarounde}.
\hfill $\square$
\end{proof}


\begin{proof}[of Theorem~\ref{thm:IMMGMS} (Sketch)]
By the argument given in~\cite{Chen18}, when combining Lemma \ref{lem:hill} and Lemma \ref{lem:sample}, 
	we should first take a union bound for $|\cR|$ going through $\lambda^*(\ell+\gamma)/\OPT$ to $\lambda^*(\ell+\gamma)$, and for each fixed length $\cR$, we apply Lemma \ref{lem:hill}
	(with $\ell$ set to $\ell+\gamma+\log 2/ \log n$).
This would properly show that with probability at most $1/n^\ell$, the $\cR$ returned by the 
	{\Sampling} procedure will not lead to an output of Algorithm~\ref{alg:hill_climbing_delta}
	as a $(1-1/e-\varepsilon)$-approximate solution to the LIM problem.

%
For time complexity, 
	when $q_{v,j}$'s are such that the optimal solution is at least as good as the best single node influence spread, we can have the inequality $\EPT\leq m\cdot \OPT / n$, where $\EPT$ is the
	expected number of incoming edges pointing to nodes in a random RR set~\cite{tang14}. 
By Lemma~\ref{lem:timeDelta} and an analysis similar to~\cite{tang15}, we can show that
	the total expected running time is bounded by:
\begin{align*}
&O\left(k\delta^{-1}(\max_{v\in V}|S_v|)(\EPT+1) \cdot \frac{\lambda^*(\ell+\gamma+\log 2/\log n)}{\OPT} \right) \\
	& =O\left( k\delta^{-1}(\max_{v\in V}|S_v|)\cdot \lambda^*(\ell) \cdot  (n+m)  \right)\\
	&=O\left ( k\delta^{-1}(\max_{v\in V}|S_v|) (M+\ell\log n)
(n+m) /  \varepsilon^2 \right).
\end{align*}
In the second inequality, besides applying $\EPT\leq m\cdot \OPT / n$, we also ignores 
	$\gamma$ and $\log 2/ \log n$, because asymptotically they are all constants.
\hfill $\square$
\end{proof}

\section{Remaining Part of the Proof of Theorem~\ref{thm:IMMGMSVSN}} \label{app:proofThm3}

We now give the additional details need to prove Theorem~\ref{thm:IMMGMSVSN}.
The main thing we want to clarify is the impact that we use a binary search for 
	the reverse sampling in the LT model part from each real node back to each 
	strategy's virtual node.
To do so, we need to reformulate a previous result $n\cdot \EPT\leq m\cdot \OPT$ in a more general setting.
Let $d'_u$ be the time needed for one-step reverse sampling from node $u$ (previously this would
	be simply the in-degree of $u$).
Given an RR set $R$, let $\omega'(R) = \sum_{u\in R} d'_u$.
Let $\EPT' = \E[\omega'(R)]$, and $\EPT'$ is the expected running time to generate one RR set.
Let $\tilde{v}$ be a random real node sampled from $V$ with probability proportional to $d'_v$'s.
Then we have
\begin{lemma}
\label{lem:generalEPT} 
$n\cdot \EPT' = \sum_{u} d'_u \cdot \E[\sigma(\{\tilde{v}\})]$. 
\end{lemma}
\begin{proof}[Sketch]
The proof essentially follows the proof of Lemma 4 in~\cite{tang14}, 
	but we need to replace the incoming edges of a node $u$ in that proof to
	$d'_u$ virtual elements of $u$, so that $d'_u$ matches with the in-degree $d_u$ of $u$.
%
%
%
\hfill $\square$
\end{proof}
Note that $\sigma(\{\tilde{v}\})$ defined in the above lemma refers to the classical influence
	spread of $\tilde{v}$ in the original graph.
We are now ready to proof Theorem~\ref{thm:IMMGMSVSN}.	

\begin{proof}[of Theorem~\ref{thm:IMMGMSVSN} (Sketch)]
The approximation ratio is ensured by Theorem~\ref{thm:VSNreduction} and the correctness of
	the {\IMM} algorithm.
For the time complexity, due to our adaption of {\IMM}, the running time is better than the one
	obtained by simply plugging in the number of nodes $n+ k\delta^{-1}d$ and the number of edges
	$m + k\delta^{-1} \sum_{v\in V} |S_v|$ into the running time formula of {\IMM}.
The analysis follows the same structure as that of {\IMM}, and we sketch the main part below.

For the greedy {\NodeSelection} procedure, given a sequence of RR sets $\cR$ of $G_A$ as input,
	its running time is $O(\sum_{R\in \cR} |R\cap U|)$.
The term $|R\cap U|$ is because 
	we only use virtual nodes as seeds and thus only the virtual nodes in an RR sets play a role
	in the {\NodeSelection} algorithm.
In fact, we could define an RR set in this case to only contain virtual nodes, but for the convenience
	of analyzing the running time, we still keep real nodes in the RR sets.
From the analysis in~\cite{tang15,Chen18}, we know that the total expected running time from all calls
	to {\NodeSelection} is 
	$O(\E[\theta]\cdot \E[|R\cap U|])$, where
	$\theta$ is the total number of RR sets generated by the algorithm.
Similarly, the time spent on generating all RR sets is $O(\E[\theta] \cdot \E[\omega'(R)] )$.
Since $\omega'(R)$ is the running time of generating $R$, we have
	$|R\cap U| \le \omega'(R)$.
Therefore, the total expected running time of the algorithm is 
	$O(\E[\theta] \cdot \E[\omega'(R)] ) =O(\E[\theta] \cdot \EPT')$.

By Lemma~\ref{lem:generalEPT}, and the assumption that the optimal solution of the LIM is at least
	as large as the optimal single node influence spread, we have
	$\EPT' \le m' \cdot \OPT / n$.
From \cite{tang15} we know that $\E[\theta] = O(\lambda^* / \OPT)$.
By Eq.~\eqref{eq:lambda*}
	$\lambda^* = O(M + \ell \log n)$.
Finally $m' = \sum_{v\in V} d'_v = O(m + \log(k\delta^{-1}) \sum_{v\in V} |S_v|)$, because for the original graph
	the reverse sampling via the triggering set uses time proportional to the in-degree of $v$
	in the original graph, and for the virtual nodes, the reverse sampling from each real node
	to each strategy's virtual nodes takes $O(\log(k\delta^{-1}))$ time via a binary search.
Combining all the above together, we know that the expected running time is
	$O(\E[\theta] \cdot \EPT') = (M + \ell \log n)(m + \log(k\delta^{-1})\sum_{v\in V} |S_v|) /\varepsilon^2)$.
\hfill $\square$
\end{proof}

\end{document}